\documentclass[runningheads,anonymous]{llncs}

\usepackage{amsmath}
\usepackage{amssymb}
\usepackage{booktabs}
\usepackage[capitalize]{cleveref}
\usepackage{subcaption}
\usepackage{url}
\usepackage[T1]{fontenc}
\usepackage{graphicx}
\usepackage{color}

\begin{document}

\title{Fast Construction of Learned Count-Min Sketch via Ternary Search (Full Version)}

\titlerunning{Fast Construction of Learned Count-Min Sketch via Ternary Search}

\author{Ryusuke Inami, Yusuke Matsui}

\authorrunning{Ryusuke Inami et al.}

\institute{The University of Tokyo, Japan\\
\email{\{inami, matsui\}@hal.t.u-tokyo.ac.jp}
}

\maketitle
\begin{abstract}
The Learned Count-Min Sketch (LCMS) is a learned data structure that estimates element frequencies in a multiset and has been experimentally shown to outperform classical data structures in the capacity-accuracy trade-off.
However, its performance depends heavily on parameter selection.
Because systematic optimization has not been adequately discussed, previous approaches relied on inefficient brute-force methods.
In this study, we propose a method to rapidly optimize the parameters of the original LCMS.
We experimentally confirmed that when the machine learning model performs well enough, using a single hash function is sufficient to optimize the weighted error metric.
Based on this, we introduce a ternary search approach to efficiently find the optimal proportion of Unique Buckets.
Our method achieves the same performance as brute-force approaches while speeding up parameter optimization by $216$-$729$ times when the machine learning model's performance is sufficient.
Furthermore, even when the model's performance is suboptimal, our approach still determines appropriate parameters $221$-$740$ times faster than the brute-force approach.

\end{abstract}

\section{Introduction}
The frequency estimation problem, which determines the frequency of each element in a multiset, is fundamental in computer science and is related to tasks across fields such as networks~\cite{liu2016one} and machine learning~\cite{aghazadeh2018mission}.
Obviously, it is possible to accurately count frequencies using dictionaries or arrays for small datasets. 
For large datasets, however, probabilistic data structures such as the Count-Min Sketch (CMS)~\cite{cormode2005improved} are becoming mainstream for estimating approximate frequencies with small capacity.

In recent years, a framework called ``learned data structures'' has been drawing attention~\cite{mitzenmacher2022algorithms}. This framework uses a machine learning model that learns the data distribution, then incorporates the model's predictions into data structures or algorithms to improve performance.
This framework has also been applied to the CMS. Hsu et al. (2019) proposed the Learned Count-Min Sketch (LCMS)~\cite{hsu2019learning}, which incorporates a machine learning model into the CMS. Since then, various types of LCMS have been proposed~\cite{aamand2023improved,cheng2023alsketch}. These works show a better trade-off between capacity and accuracy than the classical CMS.

Although LCMS outperforms the CMS, its performance depends on the parameters.
However, the question of how to optimize these parameters has not been sufficiently discussed.
Consequently, Hsu et al. (2019) \cite{hsu2019learning}, for example, determined parameter values by simulating multiple values.
This brute-force approach is time-consuming in large datasets.
Similarly, Cheng et al. (2023) \cite{cheng2023alsketch} experimentally determined the number of hash functions and the proportion of Unique Buckets, whereas Aamand et al. (2023) \cite{aamand2023improved} only verified the case where the proportion of Unique Buckets is 0.5.

To address this problem, under reasonable assumptions, we propose a method for rapidly calculating the parameters of the original LCMS \cite{hsu2019learning} to minimize weighted error, the most common evaluation metric.
Our contributions are summarized as follows:
\begin{itemize}
    \item We propose a method to optimize two types of parameters of the original LCMS without requiring simulations under reasonable assumptions.
    \item We conducted experiments on the AOL Query Logs dataset\cite{pass2006picture}, the Wikipedia Pageviews Dataset\cite{wikimedia2026pageviews}, and the Google Books Ngram datasets \cite{michel2011quantitative}, and found that considering only a single hash function is sufficient in order to minimize weighted error.
    \item We show that, in the case of a single hash function, a ternary search can efficiently calculate the optimal proportion of Unique Buckets, when the machine learning model performs sufficiently well.
    \item Compared to brute-force methods, our approach optimizes parameters up to $216$-$729$ times faster without sacrificing accuracy when the machine learning model performs well. Even when model performance is poor, it still calculates appropriate optimization $221$-$740$ times faster.
\end{itemize}

\section{Preliminaries}
\subsubsection{Notations}
For a positive integer $n$, we define $[n]$ as $\{0, 1,\dots n-1\}$.
Let $\mathcal{S}$ be the multiset.
Without loss of generality, in the following explanation, we assume that the elements of the multiset are natural numbers from $0$ to $n-1$.
Here, $n$ is the number of unique elements in $\mathcal{S}$.
For each element $i \in [n]$ of $\mathcal{S}$, we define its count in $\mathcal{S}$ as its frequency and denote it as $f_i$.
Let $N$ be the sum of $f_i$ for all $i \in [n]$ (i.e., $N = \sum_{i=0}^{n-1} f_i = \left| \mathcal{S} \right|$).

\subsubsection{Count-Min Sketch}
Count-Min Sketch~\cite{cormode2005improved} is a probabilistic data structure that estimates the frequencies of elements in a multiset $\mathcal{S}$.
CMS estimates frequencies using $d$ hash functions $h_i: [n] \to [w]$ and a two-dimensional array $C$ of $d \times w$.
When adding $f$ elements $i$ to $\mathcal{S}$, CMS updates $C[j,h_j(i)] \gets C[j,h_j(i)] + f$ for all $j \in [d]$.
The frequency of element $i$ is estimated as $\tilde{f}_i = \min_{j} C[j,h_j(i)]$.
Hereafter, adding a tilde to a variable indicates that it is an estimated output by the data structure.

\subsubsection{Learned Count-Min Sketch}
\label{ssec:22}
The original LCMS \cite{hsu2019learning} is a data structure that places a machine learning model in front of the CMS.
The machine learning model determines whether an input element is a ``heavy hitter'', i.e., one that occurs frequently.
When the model identifies a heavy hitter, it stores the element in an array called Unique Buckets to accurately calculate its frequency.
For all other elements, the CMS estimates their frequencies.
Hereafter, LCMS means the original LCMS proposed by Hsu et al. (2019).

\subsubsection{Capacity Evaluation Metrics}
In LCMS, both Unique Buckets and CMS are arrays.
Assuming the number of bits per element in the array is constant, the capacity used for the array is proportional to the number of elements; therefore, in the following discussion, we treat capacity and the number of elements in the array as equivalent.
If the length of the array that makes up Unique Buckets is $c$, and the CMS is a two-dimensional array of $d \times w$, then the capacity $B$ of the LCMS is $c + dw$.

\subsubsection{Error Evaluation Metrics}
\label{ssc:01}
When discussing the performance of frequency estimation methods, we need to quantify estimation error.
In learning-based frequency estimation methods, the estimation error is often evaluated using the following equation.
\begin{align}
    \label{eq:01}
    \mathrm{Err} = \frac{1}{N} \sum_{i \in \mathcal{S}} \left|\tilde{f}_i - f_i\right| \cdot f_i
\end{align}
This equation weights the error of $\tilde{f}_i$ by $f_i$.
We define this evaluation metric as ``weighted error''.

In CMS, the estimation error is often evaluated by the probability that the error exceeds a certain value, which is often called ``false positive rate''.
However, since the purpose of this paper is to accelerate the construction of the original LCMS, we adopt weighted error as the error evaluation metric.

\section{Proposed Method: Parameter Determination in LCMS}
\subsection{Formulation of Parameter Determination}
Given a total capacity $B$, the performance of the LCMS is determined by the following two parameters: \textbf{the number of hash functions in the CMS $d$}, and \textbf{the proportion of Unique Buckets $c/B$}.

The determination of the proportion of Unique Buckets can be formulated as follows.
We assume that a multiset $\mathcal{S}$ and a machine learning model are given.
\begin{enumerate}
    \item We prepare a multiset $\mathcal{S}'$ for parameter determination ($\mathcal{S}'$ should have a distribution similar to the multiset $\mathcal{S}$. For example, $\mathcal{S}' \subseteq \mathcal{S}$). As mentioned before, we assume the elements of $\mathcal{S}'$ are natural numbers from $0$ to $n-1$.
    \item For elements that the machine learning model determines to be heavy hitters in $\mathcal{S}'$, we assign smaller integer labels. Specifically, we rename $i$ such that $s(0) \geq s(1) \geq \dots \geq s(n-1)$. For example, consider the case where $\mathcal{S}' = \{0, 0, 1, 2, 4, 5, 5, 5\}$ and $s(0) = 4, s(1) = 1, s(2) = 3, s(3) = 0, s(4) = 2, s(5) = 6$. If we rearrange $i$ in descending order of $s(i)$, we get $5, 0, 2, 4, 1, 3$. So we rename $5$ to $0$, $0$ to $1$, $2$ to $2$, $4$ to $3$, $1$ to $4$, and $3$ to $5$. In the end, $\mathcal{S}' = \{1, 1, 4, 2, 3, 0, 0, 0\} = \{0, 0, 0, 1, 1, 2, 3, 4\}$.
    \item Take a natural number $c$ and assign $0, 1, \cdots, c - 1$ to Unique Buckets, and all other values to CMS. For example, in the $\mathcal{S}$ mentioned earlier, if $c = 2$, then $0$ and $1$ are assigned to Unique Buckets, and the rest are assigned to CMS. This operation is equivalent to assigning a heavy hitter to a Unique Bucket.
    \item Take the $c$ that showed the best performance for the data used to determine the parameters, and set $c/B$ as the optimal proportion of Unique Buckets.
\end{enumerate}
From now on, we assume that $s(0) \geq s(1) \geq \dots \geq s(n-1)$ holds true unless otherwise noted.
We define an ``ideal'' machine learning model as one in which the relative order of scores matches the relative order of true frequencies.
For example, the $s(i)$ mentioned above is the score of an ideal machine learning model because the array of $i$ sorted in descending order by its score ($5, 0, 2, 4, 1, 3$) matches the array of $i$ sorted in descending order by the frequency in $\mathcal{S}'$. If the machine learning model is ideal, then $f_0 \geq f_1 \geq \dots \geq f_{n-1}$ holds.

\begin{figure}[tb]
    \centering
    \begin{minipage}[tb]{0.42\columnwidth}
        \centering
        \includegraphics[width=\columnwidth]{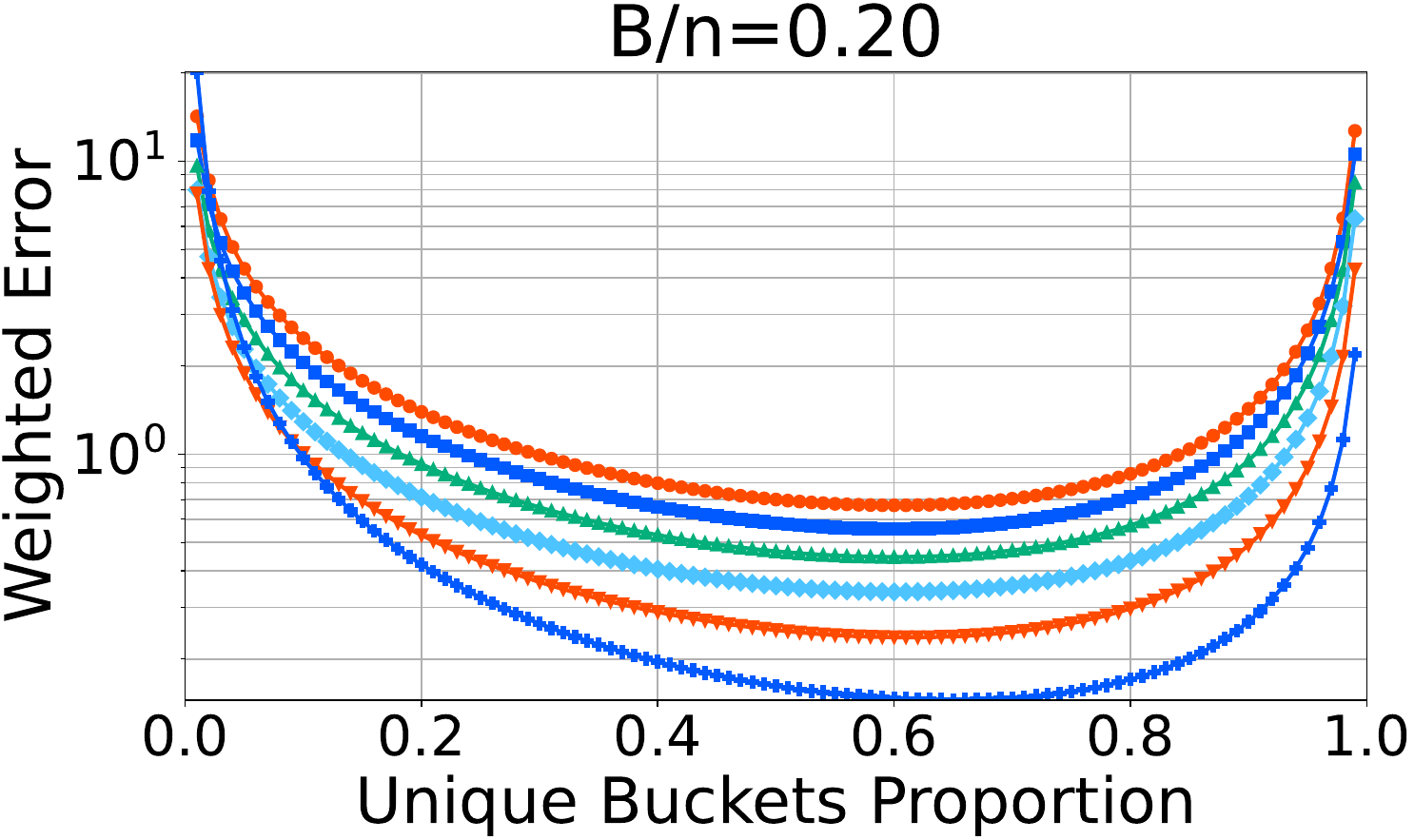}
    \end{minipage}
    \begin{minipage}[tb]{0.42\columnwidth}
        \centering
        \includegraphics[width=\columnwidth]{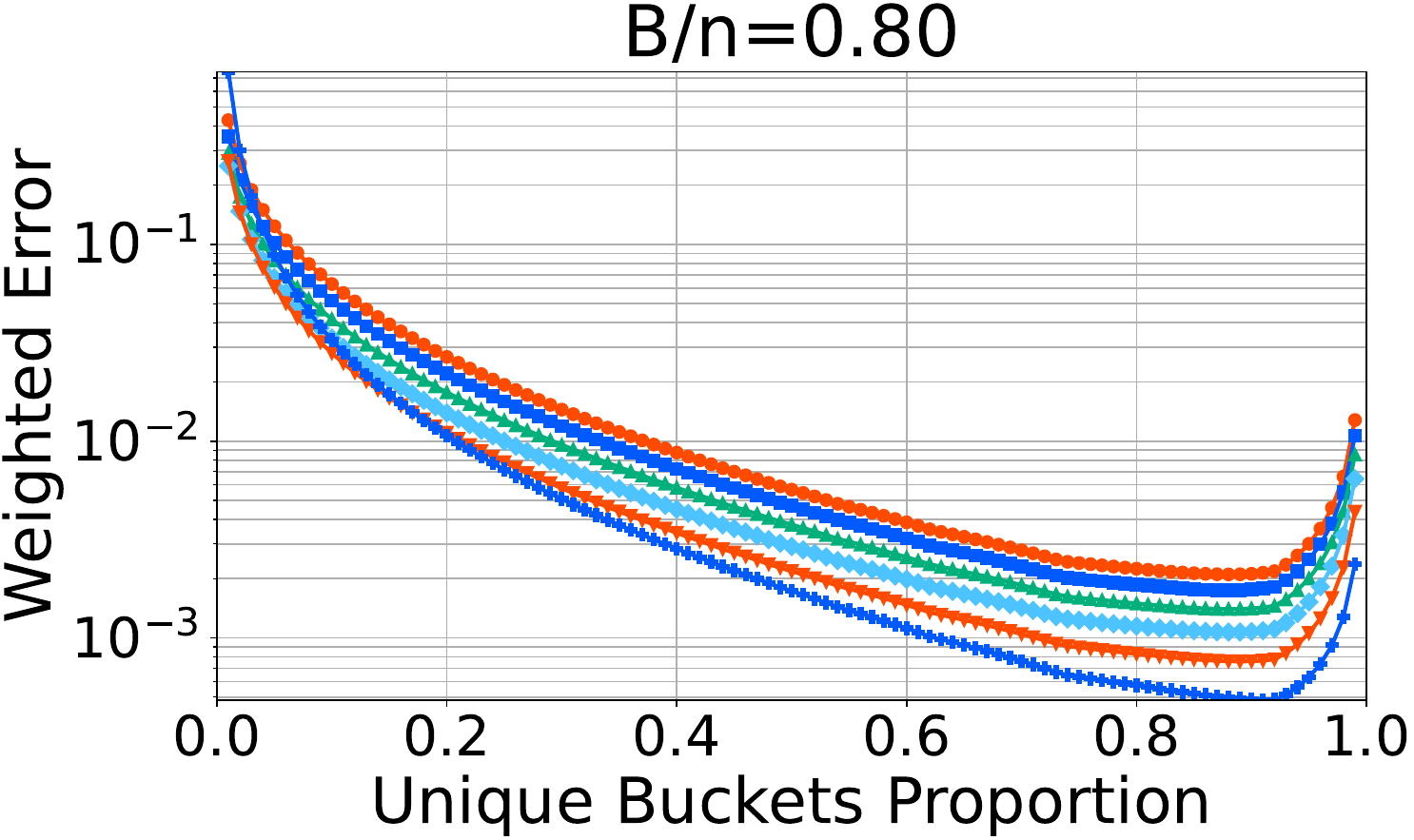}
    \end{minipage}
    \begin{minipage}[tb]{0.12\columnwidth}
        \centering
        \includegraphics[width=\columnwidth]{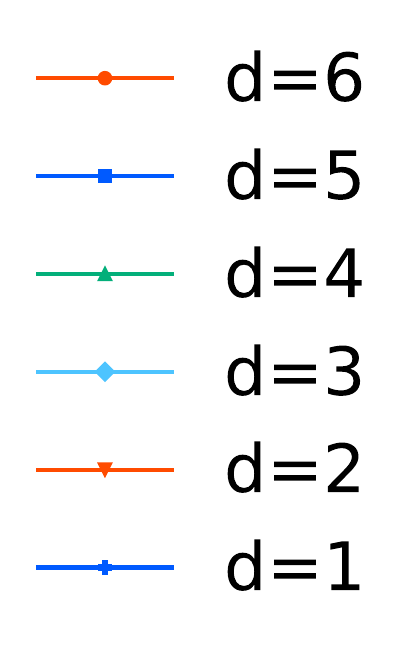}
    \end{minipage}
    \caption{The relationship between $2$ types of parameters and the weighted error on Google Books Ngram Viewer Dataset (Average over $10$ runs)}
    \label{fig:15}
\end{figure}
\subsection{Expected Error of (Learned) Count-Min Sketch}
In general, calculating the expected value of the error defined by \cref{eq:01} is difficult.
However, we can calculate it by imposing the constraint $d=1$.
The reason for limiting $d=1$ will be explained in the next subsection.

For a CMS consisting of a two-dimensional array of size $1 \times w$, the expected error is given as follows:
($h: [n] \to [w]$ is the hash function of the CMS).
\begin{align}
    &\mathbb{E} \left[ \frac{1}{N} \sum_{i = 0}^{n-1} \left|\tilde{f}_i - f_i\right| \cdot f_i \right] 
    = \mathbb{E}\left[ \frac{1}{N} \sum_{i=0}^{n-1} \left( \sum_{\substack{j \neq i \\h(j) = h(i)}} f_j \right) \cdot f_i \right]  \\
    &= \frac{1}{N} \sum_{i=0}^{n-1} \mathbb{E}\left[ \sum_{\substack{j \neq i \\h(j) = h(i)}} f_j \right] \cdot f_i 
    = \frac{1}{wN} \sum_{i=0}^{n-1} \sum_{j \neq i} f_j \cdot f_i \\
    &= \frac{1}{wN} \left( \left( \sum_{i=0}^{n-1} f_i \right)^2 - \sum_{i=0}^{n-1} f_i^2 \right) 
\end{align}

In LCMS, if the length of Unique Buckets is $c$, the elements $0, 1, \dots, c-1$ are counted accurately by Unique Buckets, resulting in no error. 
The remaining elements are where the CMS error occurs.
Therefore, the expected value of the weighted error when the length of the Unique Buckets is $c$ is expressed as follows:
\begin{align}
    \mathbb{E} \left[ \frac{1}{N} \sum_{i = c}^{n-1} \left|\tilde{f}_i - f_i\right| \cdot f_i \right]
    = \frac{1}{(w-c)N} \left( \left( \sum_{i=c}^{n-1} f_i \right)^2 - \sum_{i=c}^{n-1} f_i^2 \right) \label{eq:03}
\end{align}

\subsection{Parameter Determination in an Ideal Machine Learning Model}
\label{ssec:33}
This section describes a parameter optimization method for an ideal machine learning model.
First, let us consider the number of hash functions $d$.
We investigated the relationship between the number of hash functions and the capacity ratio of unique buckets to weighted error, using the AOL Query Logs \cite{pass2006picture}, Wikipedia Pageviews \cite{wikimedia2026pageviews}, and Google Books Ngram Viewer datasets \cite{michel2011quantitative}.
Results in the Google Books Ngram Viewer Dataset are shown in \cref{fig:15}
(Results on the other datasets are shown in the appendix \cref{apd:03}).
In all cases, the minimum weighted error is achieved when $d=1$.
From the above, we conclude that even if we limit the number of hash functions $d$ to $1$, the resulting weighted error does not deteriorate significantly.

Next, we optimize the proportion of the Unique Buckets' capacity, $c/B$.
Since $B$ is a constant, optimizing $c/B$ is equivalent to optimizing $c$, so we consider optimizing $c$ here.
Based on the findings above, we assume the number of hash functions $d$ to be $1$.
From \cref{eq:03}, the expected weighted error of the LCMS $E_c$ when $d=1$ and the elements $0, 1,\dots,c-1$ are assigned to Unique Buckets is as follows.
\begin{equation}
    E_{c} := \frac{1}{(B-c)N} \left( \left( \sum_{i=c}^{n-1} f_i \right)^2 - \sum_{i=c}^{n-1} f_i^2 \right)
\end{equation}

We consider minimizing $E_c$ as $c$ varies.
For general $f_0, f_1, \dots, f_{n-1}$, it is difficult to find the minimum value of $E_c$ more efficiently than by brute-force search.
However, the following property holds when the model is ideal (i.e., $f_0 \geq f_1 \geq \dots \geq f_{n-1}$).
\begin{theorem}
\label{thm:01}
If $f_0 \geq f_1 \geq \dots \geq f_{n-1} > 0$ holds, then $E_c$ has at most one local minimum.
\end{theorem}
The proof is provided in the appendix \cref{apd:01}.

The minimum value of a function with at most one local minimum can be calculated more efficiently using a ternary search than a brute-force search \cite{cp_algorithms_ternary_search}. 
Therefore, the optimal parameters can be calculated in $O(n + \log B)$.
More specifically, we perform a ternary search over $c$ to find the minimum value of $E_c$.
$E_c$ can be calculated in $O(1)$ by calculating the cumulative sum beforehand.
The cumulative sum can be computed in $O(n)$.
The ternary search can be done in $O(\log B)$ because $c$ satisfies $ 0 \leq c \leq B$.
Therefore, the total calculation is completed in $O(n + \log B)$.

\subsection{Parameter Determination Algorithm}
Based on the above discussion, we propose the following parameter determination method.
\begin{itemize}
    \item Determine the number of hash functions $d$ to be $1$.
    \item Calculate the proportion of Unique Buckets in $O(n+\log B)$ using the ternary search method shown in \cref{ssec:33}.
\end{itemize}

This method has the following advantages:
\begin{itemize}
    \item Parameter determination is completed rapidly because it does not involve any simulation. The method proposed by Hsu et al. (2019) evaluates various parameter combinations and selects the best one via simulation, making parameter determination time-consuming.
    \item The proposed method is equivalent to testing many Unique Buckets proportions, and for an ideal model, it selects the optimal one. In the method proposed by Hsu et al. (2019), verifying all possible parameter combinations is difficult due to the computational time required for simulation.
\end{itemize}

\section{Experiment}
\subsection{Overview}
To compare the proposed method with the brute-force method, we conducted experiments.
We compared the time required to construct LCMS using Google Books Ngram Viewer Datasets \cite{michel2011quantitative}, AOL Query Logs dataset\cite{pass2006picture}, and the Wikipedia Pageviews dataset\cite{wikimedia2026pageviews}.
We measured the time required to obtain appropriate parameters for multiple values of $B$ and evaluated performance on test data.
We conducted experiments under two conditions: an ideal machine learning model (where $f_0 \geq f_1 \geq \dots \geq f_{n-1}$ holds) and a non-ideal one; general machine learning model.
In this paper, we present only the experimental results obtained using the Google Books Ngram Viewer datasets; experiments using other datasets are described in the appendix \cref{apd:02}.

The official implementation of our method is publicly available\footnote{\url{https://github.com/lev635/Fast-LCMS}}.
Please refer to that for implementation details and more specific experimental conditions.

\subsection{Settings}
\subsubsection{Parameters}
This section describes the parameter combinations verified in both the brute-force and proposed methods.

For the brute-force method, we verified the same parameter values as in the experiment conducted by Hsu et al. (2019) \cite{hsu2019learning}.
In the experiments conducted by Hsu et al. (2019), they tested $36$ parameter combinations in total, with hash functions ranging from $1$ to $4$ and Unique Buckets ratios ranging from $0.1$ to $0.9$ in $9$ steps.\footnote{\url{https://github.com/chenyuhsu/learnedsketch}}
The brute-force method has a time complexity of $O(n)$ because it does not verify the capacity of all possible Unique Buckets.

For the proposed method, the number of hash functions was fixed to $1$, and the Unique Buckets proportion was searched using a ternary search.
In the case of an ideal machine learning model, proposed method yields results equivalent to searching all possible ratios of Unique Buckets.

\subsubsection{Capacity given to Learned Count-Min Sketch}
The capacity $B$ allocated to LCMS is set in $99$ steps, ranging from $0.01$ to $0.99$, based on the number of unique elements $n$ in $\mathcal{S}$, the dataset to be estimated.
For example, if the dataset contains $100,000$ unique elements, the $B$ to be validated will be $1000, 2000, \dots, 99000$, totaling $99$ patterns.

\subsubsection{Dataset}

The Google Books Ngram Viewer Datasets \cite{michel2011quantitative} is a dataset that aggregates word frequencies across books.
The English Version $20200217$ used in this study aggregates $1$-grams to $5$-grams for books published up to $2020$.
In this experiment, we used the $1$-gram dataset. 
We created yearly word-frequency datasets, excluding words longer than 65 characters. 
We used the $2019$ dataset ($29,959,085$ words).

\subsubsection{Machine Learning Model}
We trained a machine learning model to predict word frequency.
To predict word frequency, we adopted the architecture from Hsu et al. (2019) \cite{hsu2019learning}, comprising an embedding layer, a BiLSTM layer, and a fully connected layer. 
Due to the large size of the $2019$ dataset, we trained the model using $3,000,000$ stratified samples per epoch.

\subsection{Results}

\subsubsection{Ideal Machine Learning Model}
\begin{figure}[tb]
    \centering
    \begin{minipage}[tb]{0.48\columnwidth}
        \centering
        \includegraphics[width=\columnwidth]{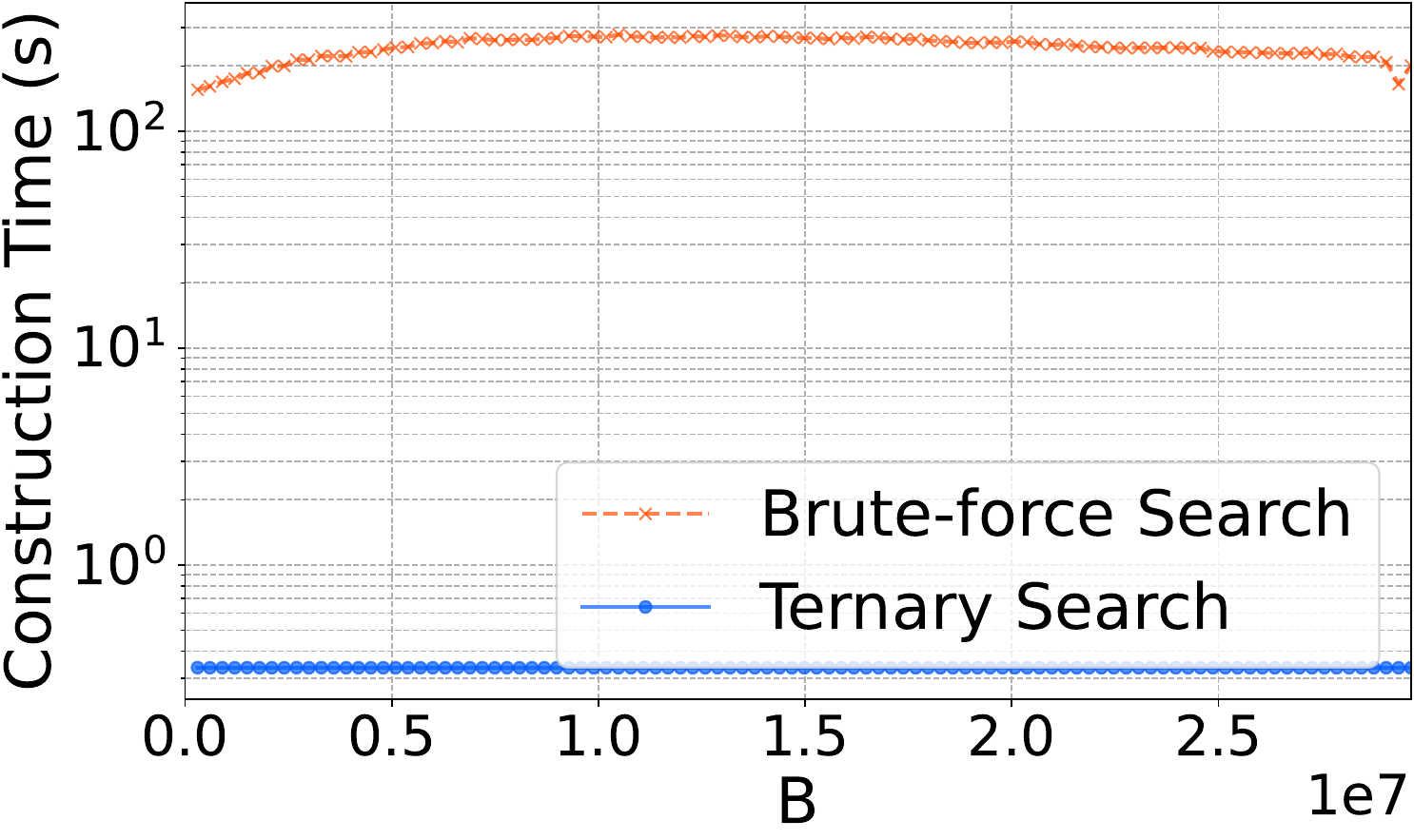}
        \subcaption{Comparison of construction times}
    \end{minipage}
    \begin{minipage}[tb]{0.48\columnwidth}
        \centering
        \includegraphics[width=\columnwidth]{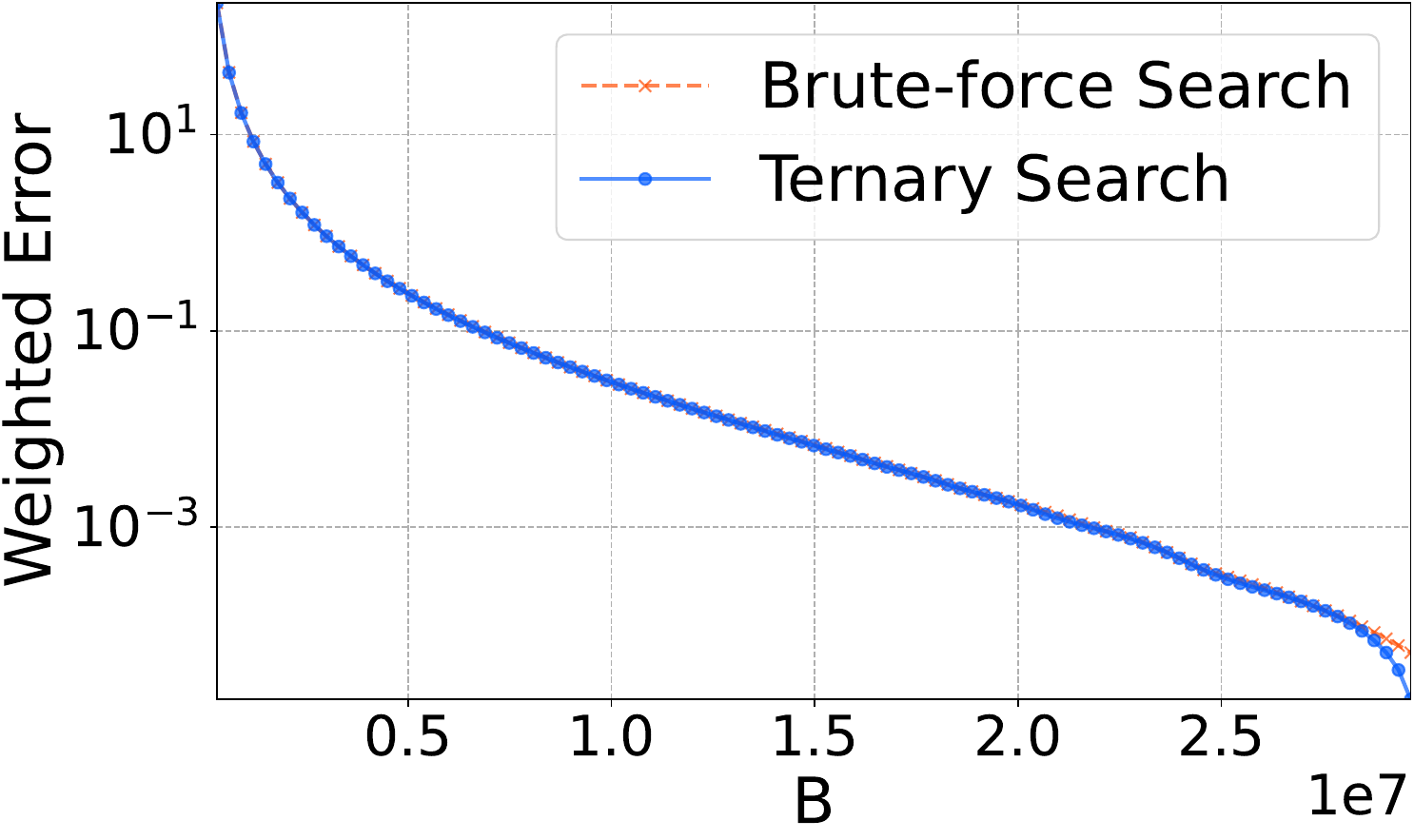}
        \subcaption{Comparison of weighted error}
    \end{minipage}
    \caption{Results on an ideal machine learning model (Google Ngram Viewer)}
    \label{fig:06}
\end{figure}

\cref{fig:06} shows the evaluation results of the proposed method using the $2019$ data under the assumption of an ideal machine learning model. 
Including pre-calculation, our approach accelerated construction time by approximately $729$ times compared to brute-force methods (from an average of $245$ seconds down to $0.335$ seconds). This speedup suggests the dataset could be further expanded. 
Furthermore, our method maintained comparable performance, improving average test data error by $2.74\%$. This confirms our assumption that considering only a $1$-hash function is sufficient.

\subsubsection{General Machine Learning Model}
\begin{figure}[tb]
    \centering
    \begin{minipage}[tb]{0.48\columnwidth}
        \centering
        \includegraphics[width=\columnwidth]{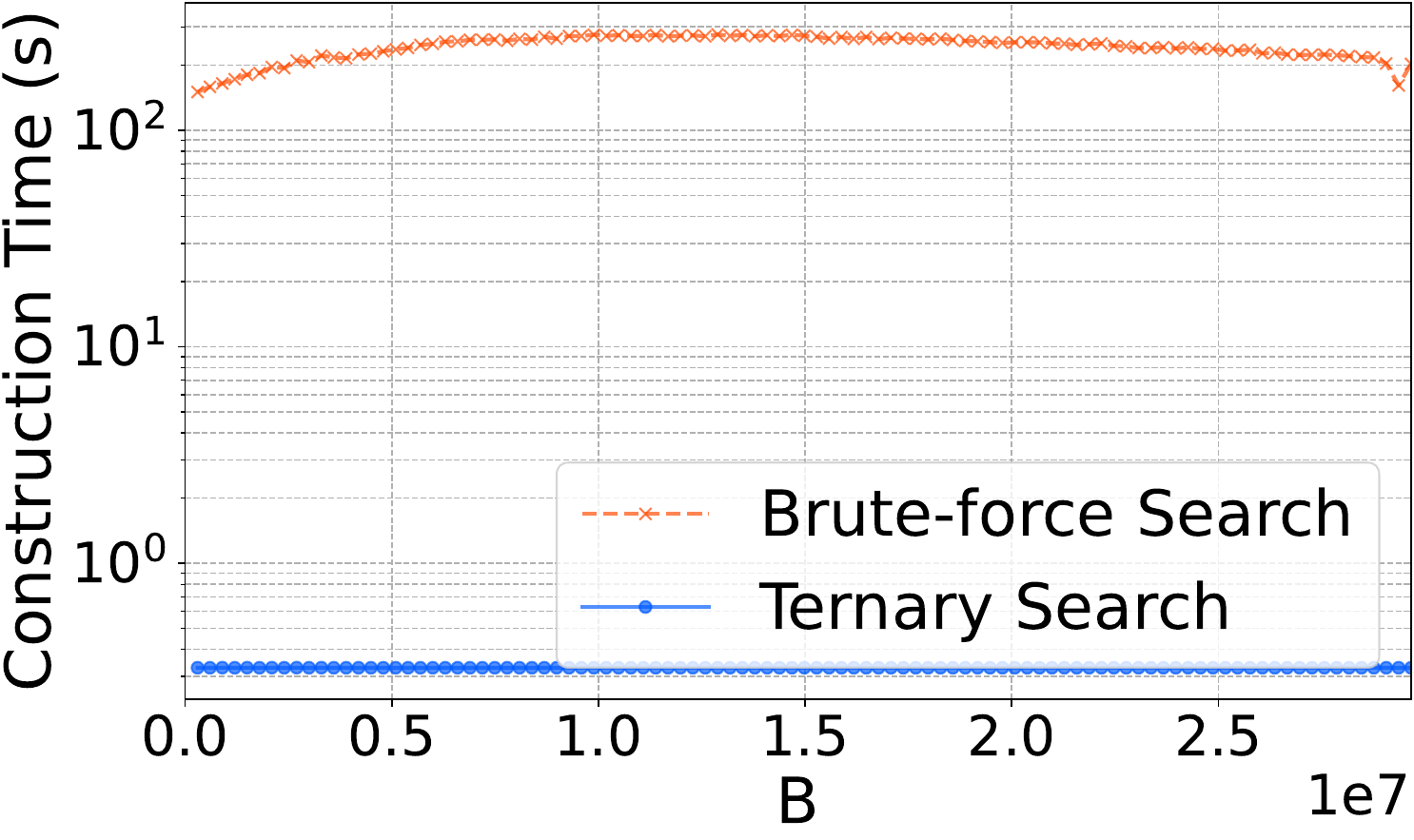}
        \subcaption{Comparison of construction times}
    \end{minipage}
    \begin{minipage}[tb]{0.48\columnwidth}
        \centering
        \includegraphics[width=\columnwidth]{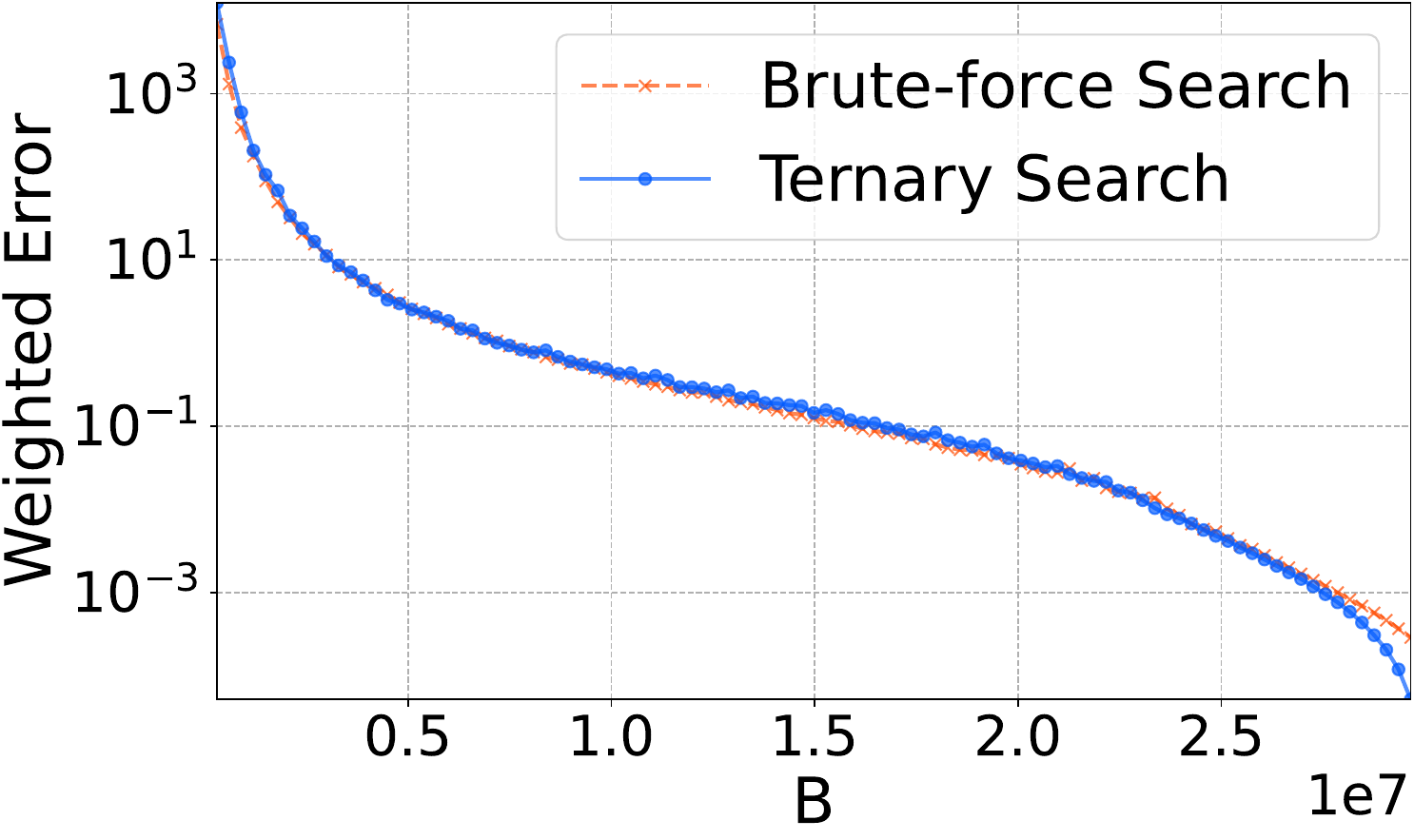}
        \subcaption{Comparison of weighted error}
    \end{minipage}
    \caption{Results on a general machine learning model (Google Ngram Viewer)}
    \label{fig:07}
\end{figure}

\cref{fig:07} presents the corresponding results of the proposed method using a general machine learning model on the $2019$ data. 
Our approach achieved roughly $740$ times faster construction than brute-force methods, while the average test-data error increased by only $5.21\%$. 
Although the expected error for general models has multiple local minima with respect to the capacity of Unique Buckets, the overall trend changes from ``decreasing → increasing''. Thus, applying ternary search poses no practical problems.

\section{Conclusion}
We proposed a ternary search method to efficiently calculate optimal parameters for the original LCMS \cite{hsu2019learning}.
For ideal machine learning models, our method achieved $216$-$729$ times faster computation and reduced the weighted error by $1.43$-$2.74\%$, validating the single-hash-function assumption and the advantage of ternary search.
For non-ideal models, the calculation remained $221$-$740$ times faster, despite an increase in error of up to $5.21\%$. 
Overall, even when strict theoretical assumptions fail in non-ideal models, our method successfully and efficiently finds appropriate parameters.

\subsubsection{Disclosure of Interests.}
The authors declare that they have no competing interests.

\bibliographystyle{splncs04}
\bibliography{bibliography}

@article{aamand2023improved,
  author = {Aamand, Anders and Chen, Justin and Nguyen, Huy and Silwal, Sandeep and Vakilian, Ali},
  title = {Improved frequency estimation algorithms with and without predictions},
  journal = {Advances in Neural Information Processing Systems},
  volume = {36},
  pages = {14387--14399},
  year = {2023}
}

@inproceedings{aghazadeh2018mission,
  author = {Aghazadeh, Amirali and Spring, Ryan and LeJeune, Daniel and Dasarathy, Gautam and Shrivastava, Anshumali and others},
  title = {Mission: ultra large-scale feature selection using count-sketches},
  booktitle = {Proceedings of the 35th International Conference on Machine Learning (ICML)},
  pages = {80--88},
  publisher = {PMLR},
  year = {2018}
}

@article{cheng2023alsketch,
  author = {Cheng, Xiaojun and Jing, Xuyang and Yan, Zheng and Li, Xian and Wang, Pu and Wu, Wei},
  title = {{ALSketch}: an adaptive learning-based sketch for accurate network measurement under dynamic traffic distribution},
  journal = {Journal of Network and Computer Applications},
  volume = {216},
  pages = {103659},
  publisher = {Elsevier},
  year = {2023}
}

@article{cormode2005improved,
  author = {Cormode, Graham and Muthukrishnan, Shan},
  title = {An improved data stream summary: the count-min sketch and its applications},
  journal = {Journal of Algorithms},
  volume = {55},
  number = {1},
  pages = {58--75},
  publisher = {Elsevier},
  year = {2005}
}

@inproceedings{hsu2019learning,
  author = {Hsu, Chen-Yu and Indyk, Piotr and Katabi, Dina and Vakilian, Ali},
  title = {Learning-based frequency estimation algorithms},
  booktitle = {Proceedings of the 7th International Conference on Learning Representations (ICLR)},
  year = {2019}
}

@inproceedings{liu2016one,
  author = {Liu, Zaoxing and Manousis, Antonis and Vorsanger, Gregory and Sekar, Vyas and Braverman, Vladimir},
  title = {One sketch to rule them all: rethinking network flow monitoring with univmon},
  booktitle = {Proceedings of the ACM SIGCOMM 2016 Conference},
  pages = {101--114},
  year = {2016}
}

@article{mitzenmacher2022algorithms,
  author = {Mitzenmacher, Michael and Vassilvitskii, Sergei},
  title = {Algorithms with predictions},
  journal = {Communications of the ACM},
  volume = {65},
  number = {7},
  pages = {33--35},
  publisher = {ACM},
  year = {2022}
}

@inproceedings{pass2006picture,
  author = {Pass, Greg and Chowdhury, Abdur and Torgeson, Cayley},
  title = {A picture of search},
  booktitle = {Proceedings of the 1st International Conference on Scalable Information Systems (InfoScale '06)},
  publisher = {ACM},
  doi = {10.1145/1146847.1146848},
  year = {2006}
}

@misc{cp_algorithms_ternary_search,
  author = {{CP-Algorithms contributors}},
  title = {Ternary search},
  url = {https://cp-algorithms.com/num_methods/ternary_search.html},
  note = {Accessed: 2026-02-02},
  year = {2025}
}

@article{michel2011quantitative,
  author = {Michel, Jean-Baptiste and Shen, Yuan Kui and Aiden, Aviva Presser and Veres, Adrian and Gray, Matthew K and Google Books Team and Pickett, Joseph P and Hoiberg, Dale and Clancy, Dan and Norvig, Peter and others},
  title = {Quantitative analysis of culture using millions of digitized books},
  journal = {Science},
  volume = {331},
  number = {6014},
  pages = {176--182},
  publisher = {American Association for the Advancement of Science},
  year = {2011}
}

@misc{wikimedia2026pageviews,
  author = {{Wikimedia Foundation}},
  title = {Wikipedia pageview data dumps},
  year = {2026},
  url = {https://dumps.wikimedia.org/other/pageviews/},
  note = {Accessed: 2026-04-07}
}

\appendix
\section{Proof of \cref{thm:01}}
\label{apd:01}
\begin{proof}
If the condition where $E_{c-1} \leq E_c$ and $E_c \geq E_{c+1}$ are simultaneously satisfied is never met for any $c \in \{1, 2, \dots, B-2\}$, then $E_c$ possesses at most one local minimum. Accordingly, we proceed to prove this assertion. We employ a proof by contradiction. We assume that there exists a $c$ such that both $E_{c-1} \leq E_c$ and $E_c \geq E_{c+1}$ hold, and derive a contradiction.

We restate the definition of $E_c$:
\begin{equation}
    E_{c} := \frac{1}{(B-c)N} \left( \left( \sum_{i=c}^{n-1} f_i \right)^2 - \sum_{i=c}^{n-1} f_i^2 \right)
\end{equation}
For the sake of simplicity, we define $N_c$ and $S_c$ as follows:
\begin{align}
    N_c &:= \left( \left( \sum_{i=c}^{n-1} f_i \right)^2 - \sum_{i=c}^{n-1} f_i^2 \right) \\
    S_c &:= \sum_{i=c}^{n-1} f_i
\end{align}

For a general $c$, $E_{c+1}$ can be expressed as follows:
\begin{align}
     E_{c+1}
    &= \frac{1}{(B-c-1)N} \left[ \left( \sum_{i=c+1}^{n-1} f_i \right)^2 - \sum_{i=c+1}^{n-1} f_i^2 \right] \\
    &= \frac{1}{(B-c-1)N} \left[ \left( \sum_{i=c}^{n-1} f_i \right)^2 - \sum_{i=c}^{n-1} f_i^2 - 2 f_c \sum_{i=c+1}^{n-1} f_i \right] \\
    &= \frac{N_c - 2 f_c S_{c+1}}{(B-c-1)N}
\end{align}

Therefore,
\begin{align}
    E_{c+1} - E_c
    &= \frac{N_c - 2 (B-c) f_c S_{c+1} }{(B-c-1)(B-c)N} 
\end{align}

That is, based on the assumption, both of the following hold:
\begin{align}
    E_{c-1} \leq E_c &\Leftrightarrow N_{c-1} - 2(B - c + 1) f_{c-1} S_c \geq 0 \label{eq:81} \\
    E_{c} \geq E_{c+1} &\Leftrightarrow N_c - 2(B - c) f_c S_{c+1} \label{eq:82} \leq 0
\end{align}

We further manipulate \cref{eq:81}.
\begin{align}
    & N_{c-1} - 2(B - c + 1) f_{c-1} S_c \geq 0 \\
    &\Leftrightarrow N_c - 2 (B - c) f_{c-1} S_c \geq 0 \label{eq:83}
\end{align}

Since \cref{eq:82} and \cref{eq:83} hold simultaneously,
\begin{align}
    2 (B - c) f_{c-1} S_c \leq 2 (B - c) f_c S_{c+1}
    &\Leftrightarrow f_{c-1} S_c \leq f_c S_{c+1} \\
    &\Leftrightarrow f_{c-1} (f_c + S_{c+1}) \leq f_c S_{c+1} \\
    &\Leftrightarrow f_{c-1} f_c \leq S_c (f_c - f_{c-1}) \label{eq:84}
\end{align}

In \cref{eq:84}, the left-hand side is strictly positive since $f_{c-1} > 0$ and $f_c > 0$, whereas the right-hand side is non-positive because $f_c \leq f_{c-1}$.This yields a contradiction. Thus, the theorem is proven.

\end{proof}

\section{Experimental Results on Other Datasets}
\label{apd:02}
\subsection{AOL Query Logs dataset}
\subsubsection{Dataset}
AOL Query Log dataset\cite{pass2006picture} is a dataset that aggregates $21$ million searches performed by $650,000$ users over a $3$-month period.
Hsu et al.'s experiment used a dataset that aggregated word-search relationships daily; this study also used a similar dataset.
The experiment used data for $92$ days, from day $0$ to day $91$.
The average number of unique elements over the $92$-day period is approximately $170000$.
We used data from day $0$ to day $6$ to train the machine learning models, and data from day $20$ and day $50$ to evaluate the methods.

\subsubsection{Machine Learning Model}
For the machine learning model, we used a structure similar to that of Hsu et al. (2019) \cite{hsu2019learning}.
The model estimates the number of times a search term has been searched, taking it as input. It consists of an RNN using LSTM cells and a fully connected layer.
To input the search term into the RNN, we need to generate a character embedding using the same one as in Hsu et al. (2019).

\subsubsection{Results}
\cref{fig:08} and \cref{fig:09} show the results of our proposed method using the ideal and general machine learning models, respectively. 
Compared to brute-force methods, our approach achieved roughly $216$x and $221$x faster construction times, while improving average test errors by $2.09\%$ and $2.39\%$ with almost the same performance.

\begin{figure}[tb]
    \centering
    \begin{minipage}[tb]{0.48\columnwidth}
        \centering
        \includegraphics[width=\columnwidth]{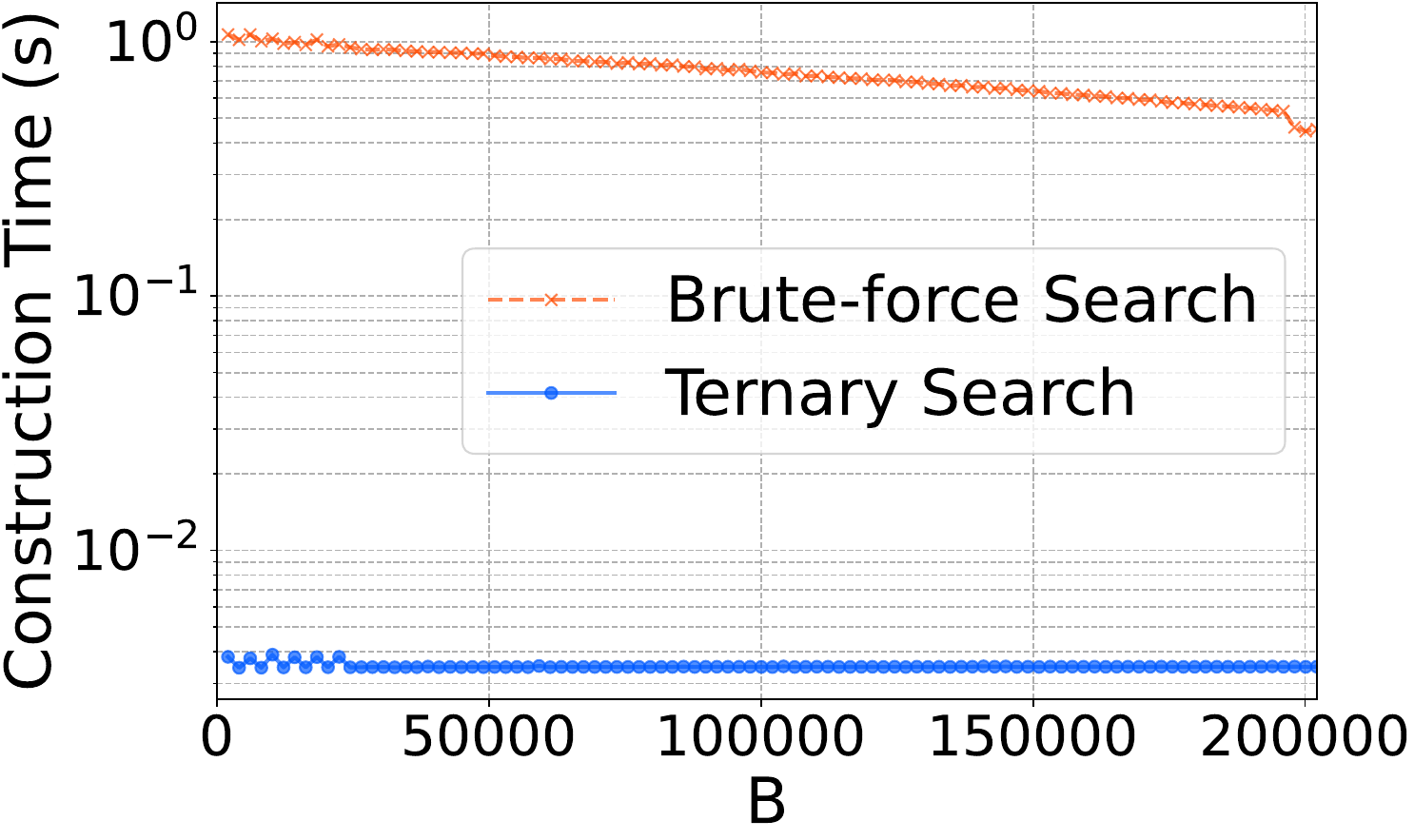}
        \subcaption{Comparison of construction times}
    \end{minipage}
    \begin{minipage}[tb]{0.48\columnwidth}
        \centering
        \includegraphics[width=\columnwidth]{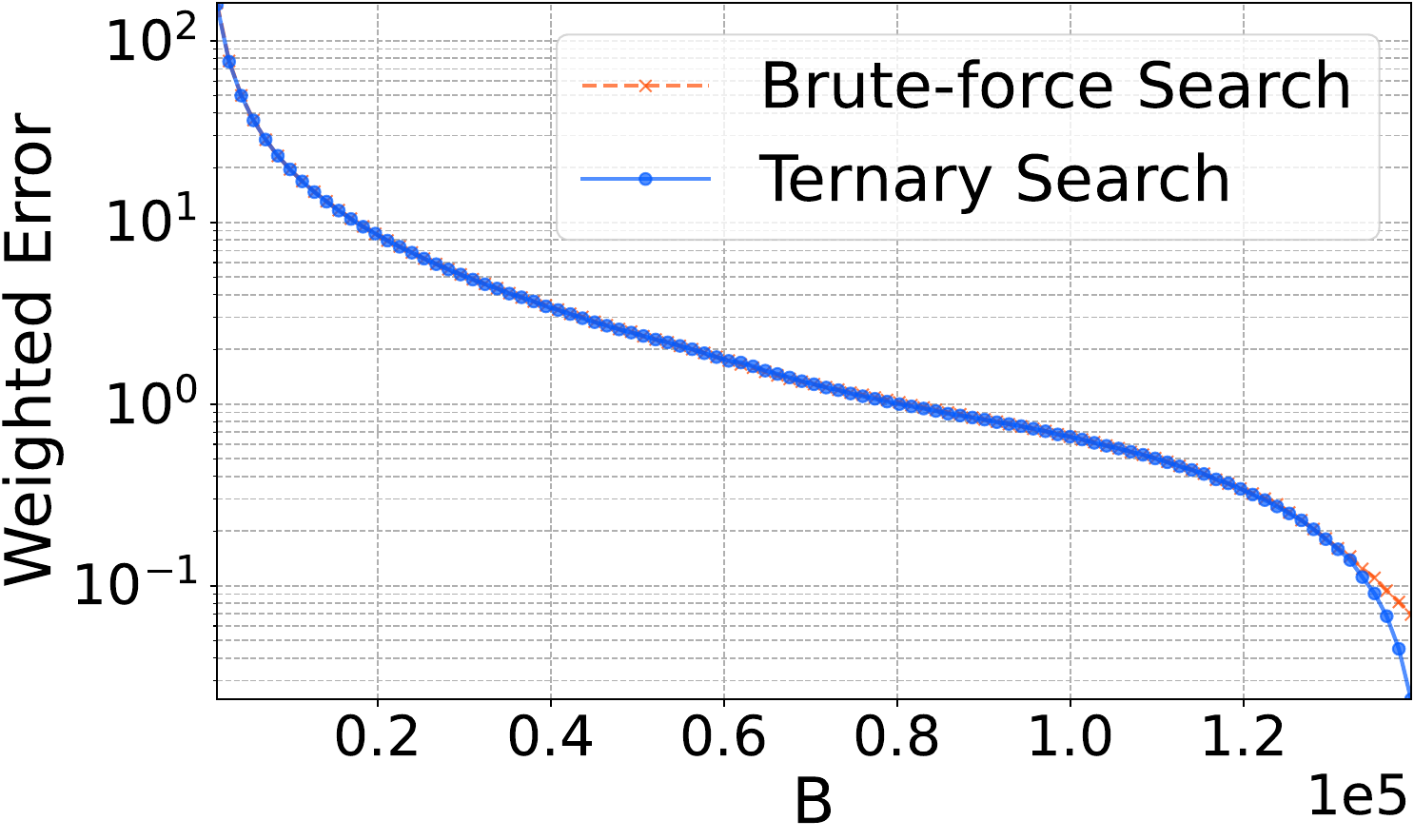}
        \subcaption{Comparison of weighted error}
    \end{minipage}
    \caption{Results on an ideal machine learning model (AOL Query Logs)}
    \label{fig:08}
\end{figure}
\begin{figure}[tb]
    \centering
    \begin{minipage}[tb]{0.48\columnwidth}
        \centering
        \includegraphics[width=\columnwidth]{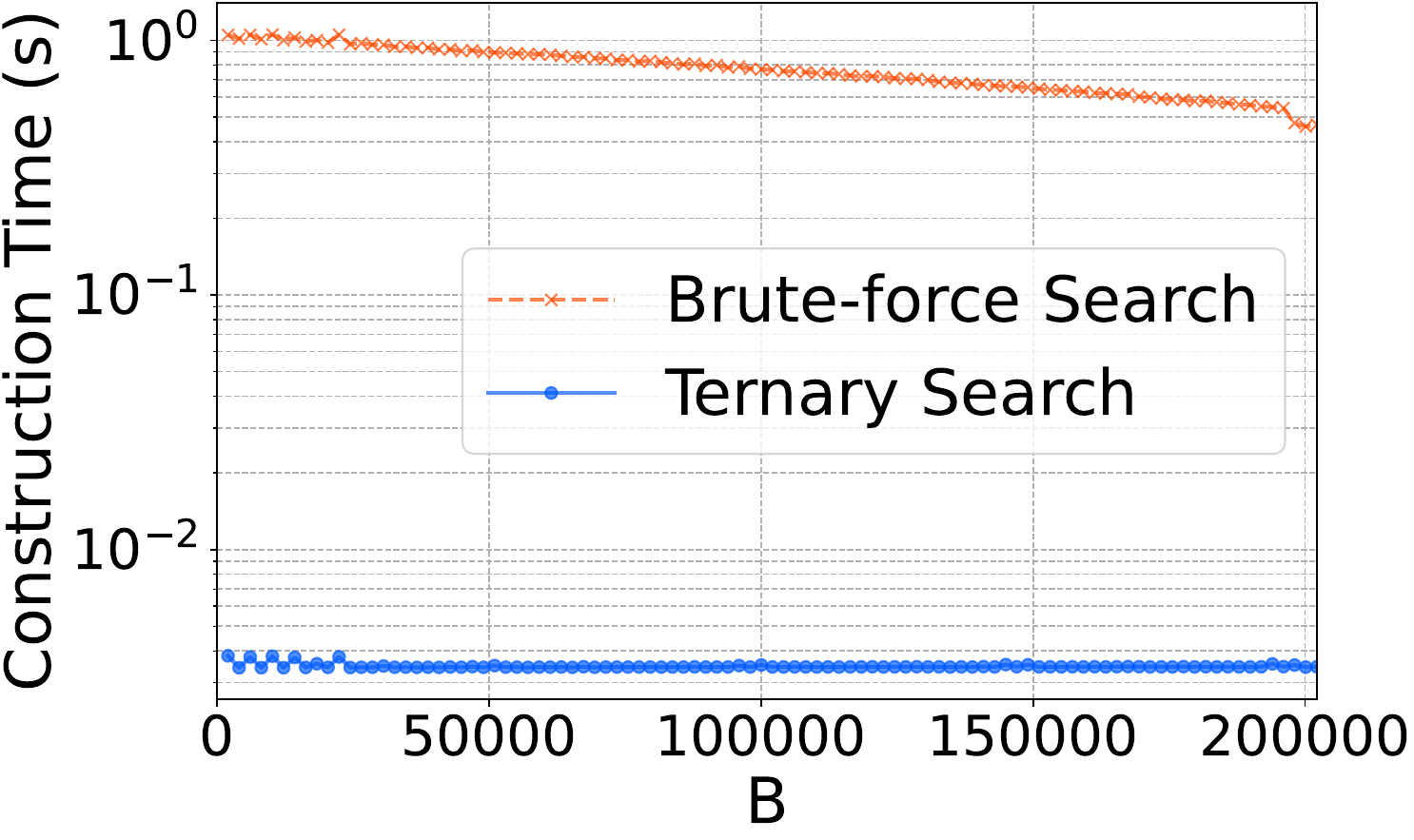}
        \subcaption{Comparison of construction times}
    \end{minipage}
    \begin{minipage}[tb]{0.48\columnwidth}
        \centering
        \includegraphics[width=\columnwidth]{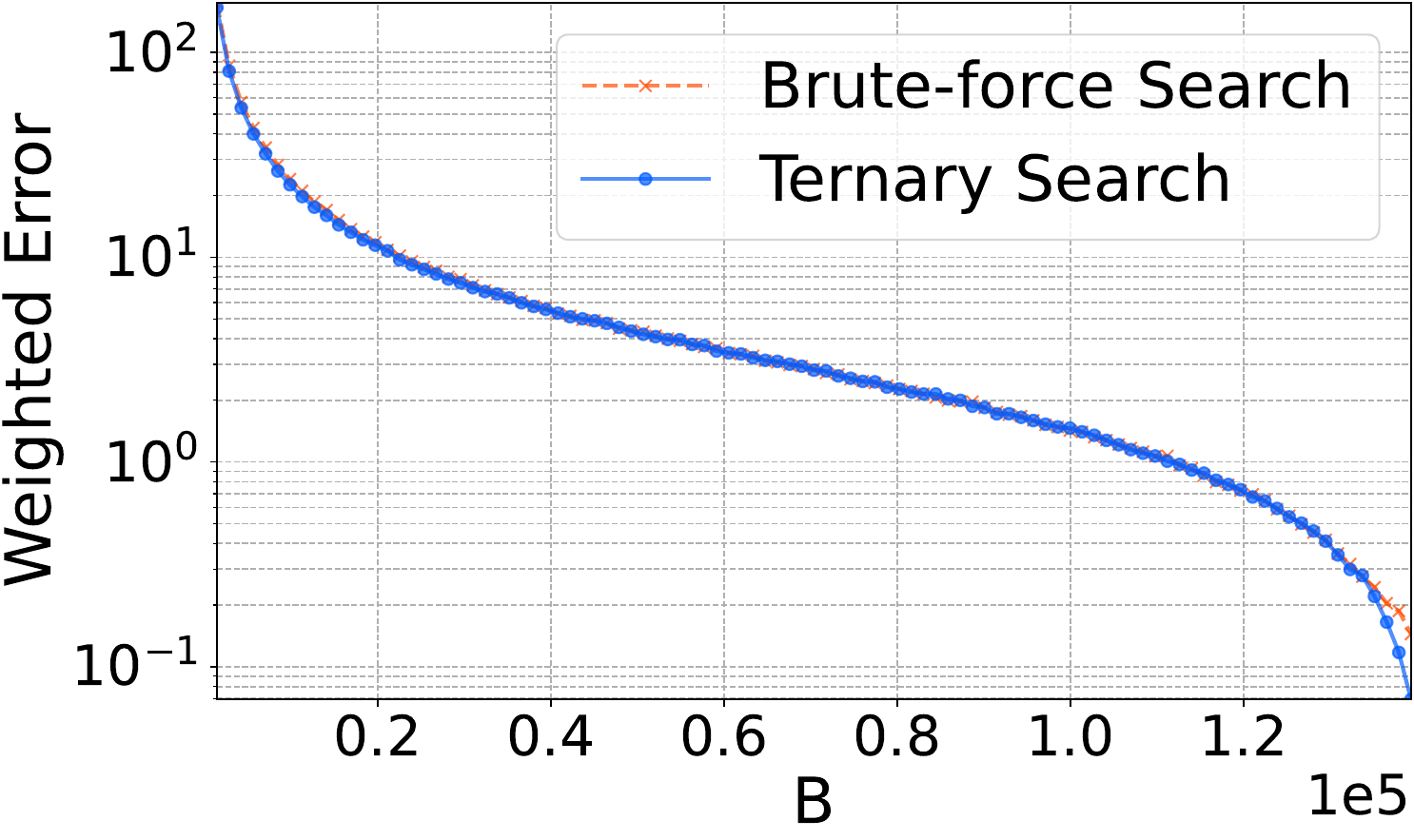}
        \subcaption{Comparison of weighted error}
    \end{minipage}
    \caption{Results on a general machine learning model (AOL Query Logs)}
    \label{fig:09}
\end{figure}

\subsection{Wikipedia Pageviews Dataset}
\subsubsection{Dataset}
The Wikipedia Pageviews Dataset \cite{wikimedia2026pageviews} records the number of page views for each Wikipedia page from $2015$ to the present.
The number of views is aggregated every hour, with one dataset per hour.

In this study, only pages where the \texttt{page\_title} consisted of ASCII characters and the \texttt{domain\_code} was \texttt{en} were used in the experiment.
In this study, we used datasets from January $1$ and $2$, $2020$.

\subsubsection{Machine Learning Model}
We trained a machine learning model to predict page views from the page title.
Based on the machine learning model structure adopted by Hsu et al. (2019) \cite{hsu2019learning}, a model comprising an embedding layer, an LSTM layer, and a fully connected layer was used.
Training was performed using probability sampling weighted by page-view counts.

\subsubsection{Results}
\begin{figure}[tb]
    \centering
    \begin{minipage}[tb]{0.48\columnwidth}
        \centering
        \includegraphics[width=\columnwidth]{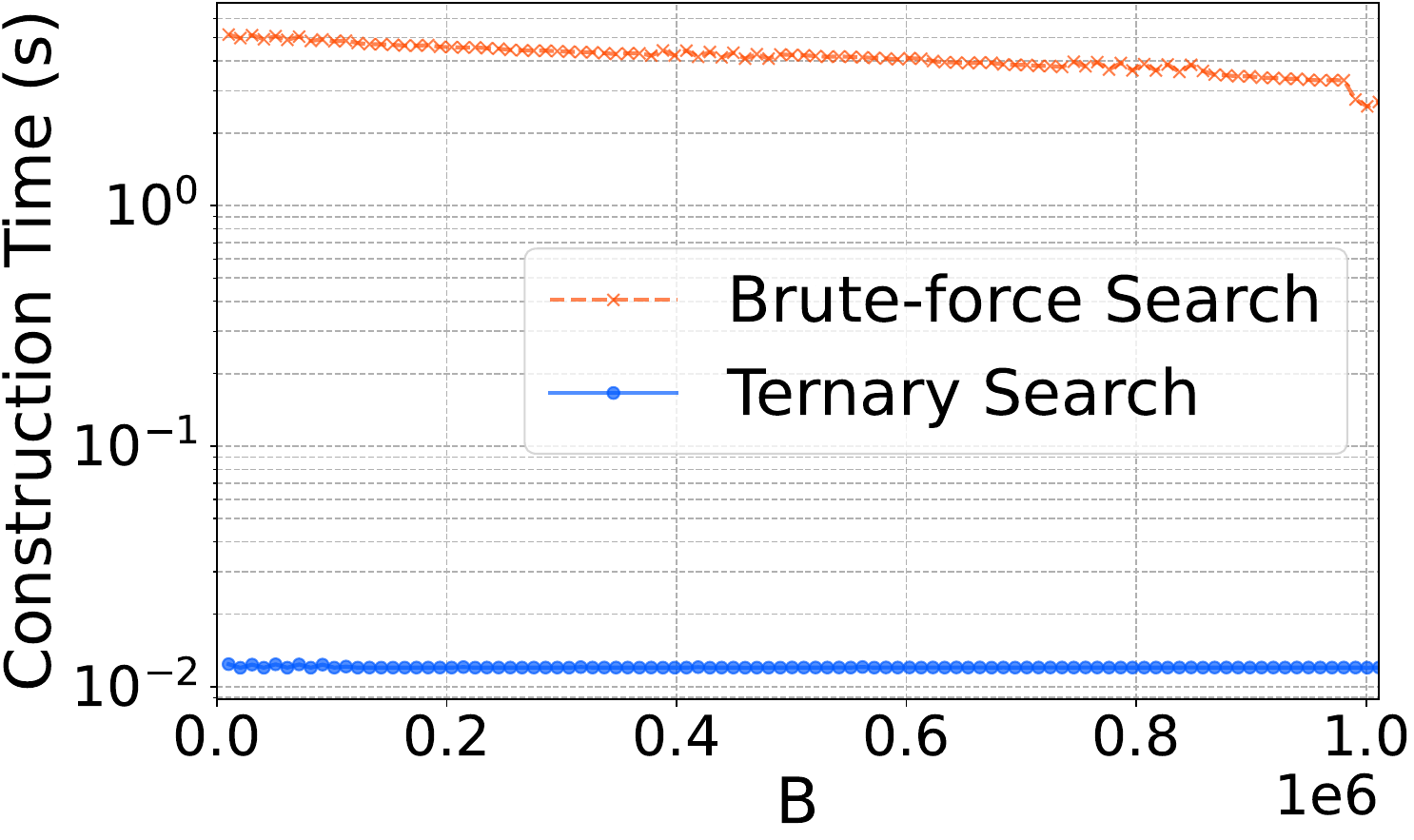}
        \subcaption{Comparison of construction times}
    \end{minipage}
    \begin{minipage}[tb]{0.48\columnwidth}
        \centering
        \includegraphics[width=\columnwidth]{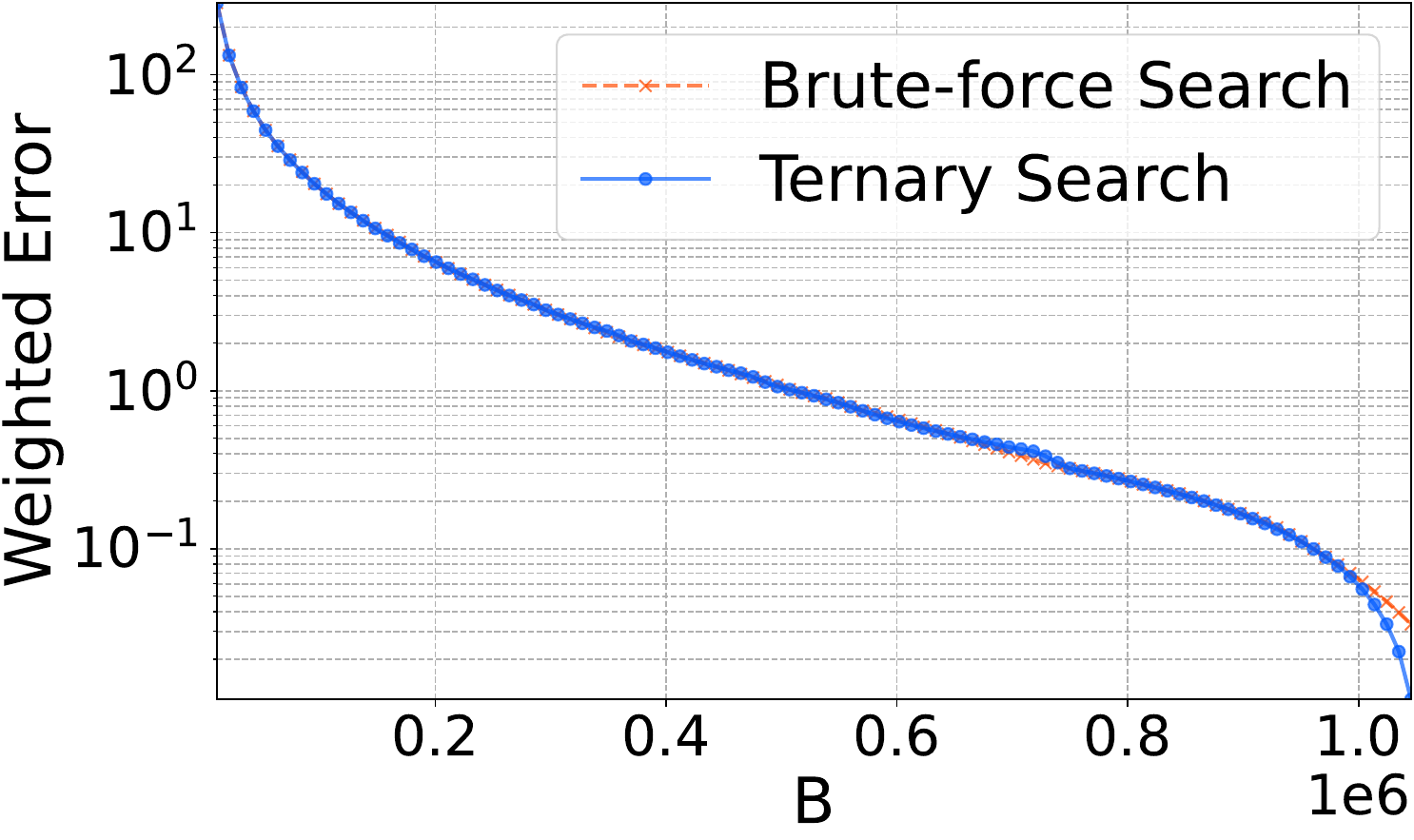}
        \subcaption{Comparison of weighted error}
    \end{minipage}
    \caption{Results on an ideal machine learning model (Wikipedia Pageviews)}
    \label{fig:10}
\end{figure}
\begin{figure}[tb]
    \centering
    \begin{minipage}[tb]{0.48\columnwidth}
        \centering
        \includegraphics[width=\columnwidth]{figures/lcms_construction_speed/wikipedia_ideal_time.pdf}
        \subcaption{Comparison of construction times}
    \end{minipage}
    \begin{minipage}[tb]{0.48\columnwidth}
        \centering
        \includegraphics[width=\columnwidth]{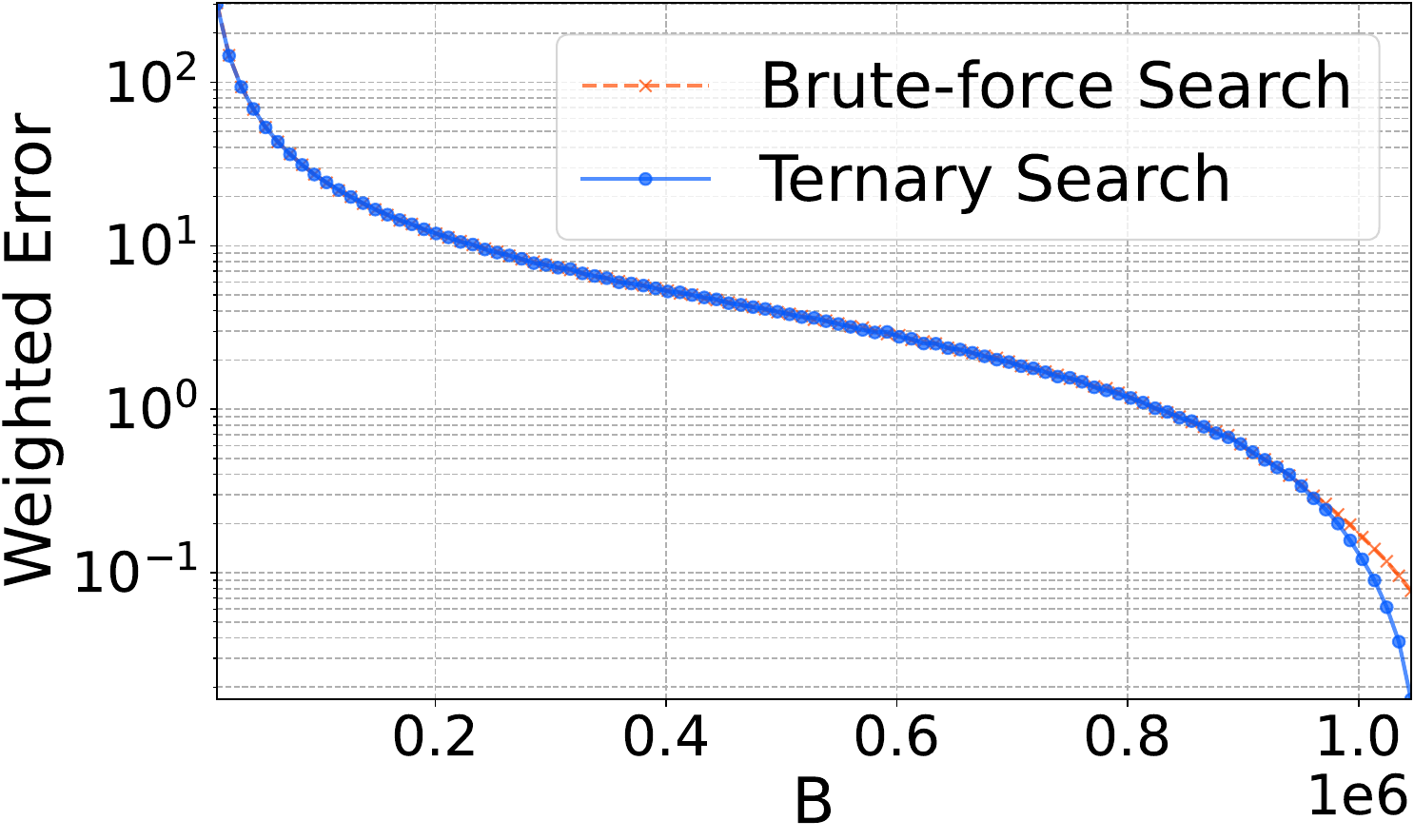}
        \subcaption{Comparison of weighted error}
    \end{minipage}
    \caption{Results on a general machine learning model (Wikipedia Pageviews)}
    \label{fig:11}
\end{figure}

\cref{fig:10} and \cref{fig:11} show the evaluation results of our proposed method using the ideal and general machine learning models, respectively. 
Compared to brute-force methods, our approach achieved roughly $343$x and $390$x faster construction times, and improved test data errors by $1.43$ and $3.18$ points, respectively.

\section{The Relationship between $2$ types of Parameters and Weighted Error}
\label{apd:03}
\subsection{AOL Query Logs dataset}
The experimental results obtained from the AOL Query Log dataset \cite{pass2006picture} are shown in \cref{fig:20}.
\begin{figure}[tb]
    \centering
    \begin{minipage}[tb]{0.48\columnwidth}
        \centering
        \includegraphics[width=\columnwidth]{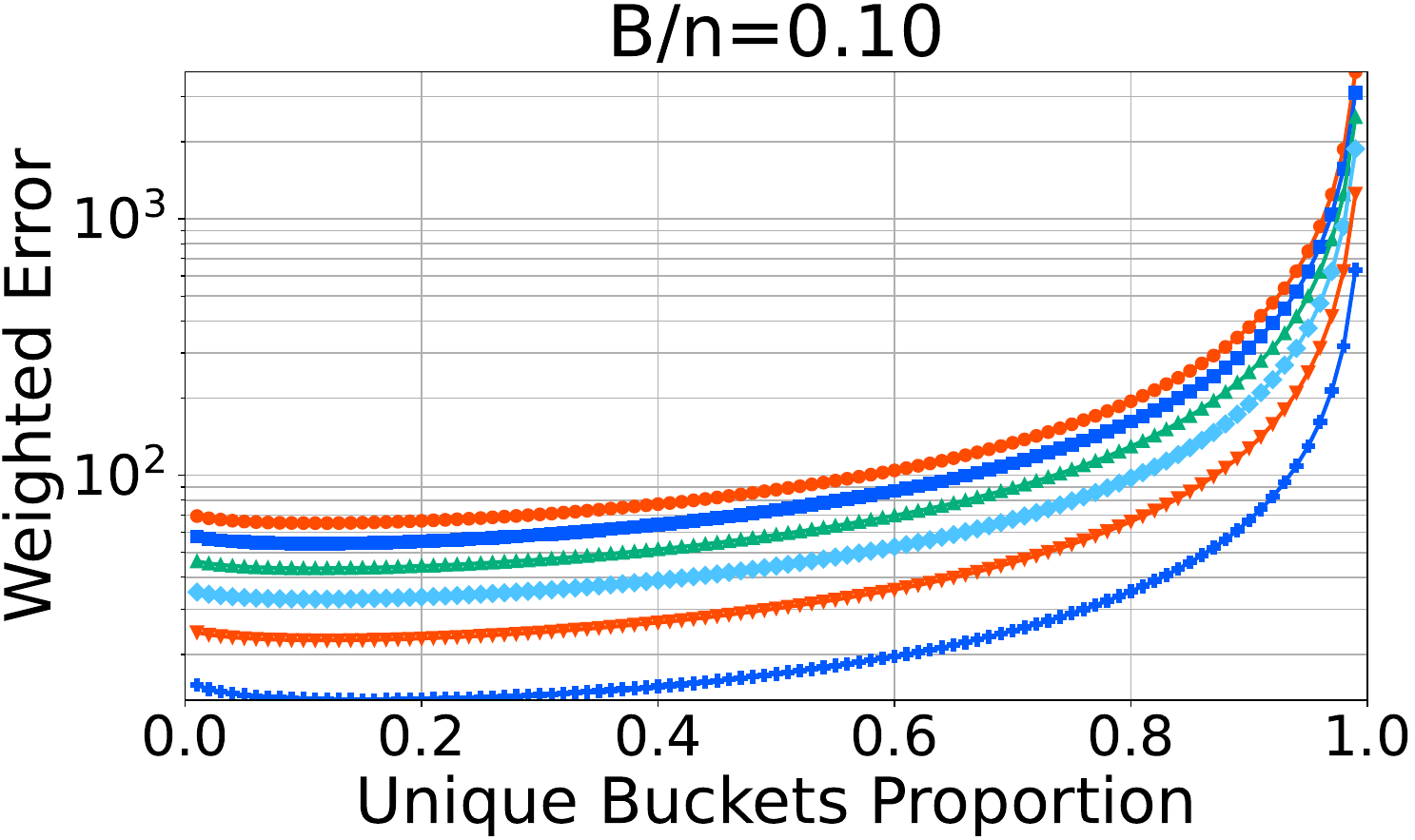}
    \end{minipage}
    \begin{minipage}[tb]{0.48\columnwidth}
        \centering
        \includegraphics[width=\columnwidth]{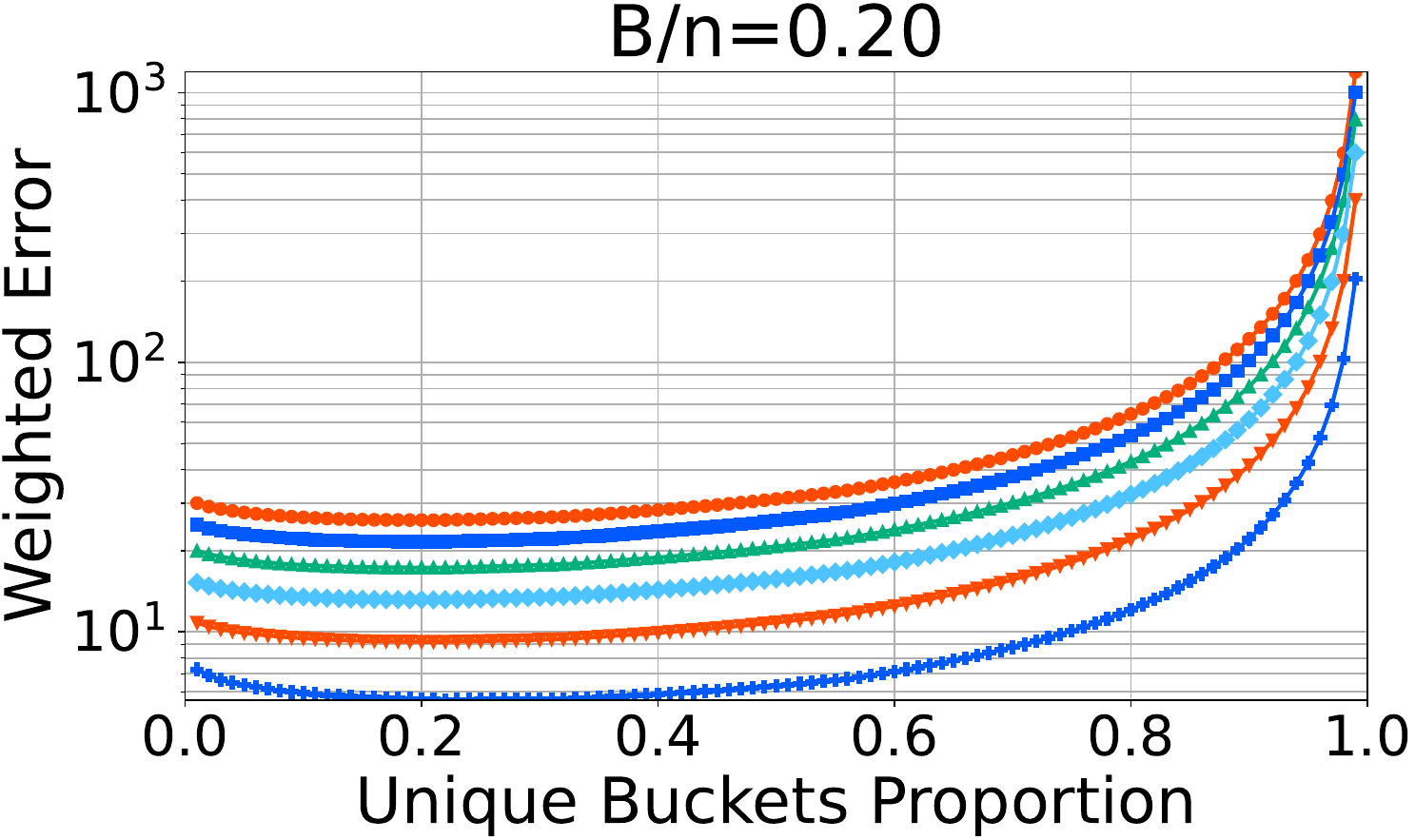}
    \end{minipage}
    \begin{minipage}[tb]{0.48\columnwidth}
        \centering
        \includegraphics[width=\columnwidth]{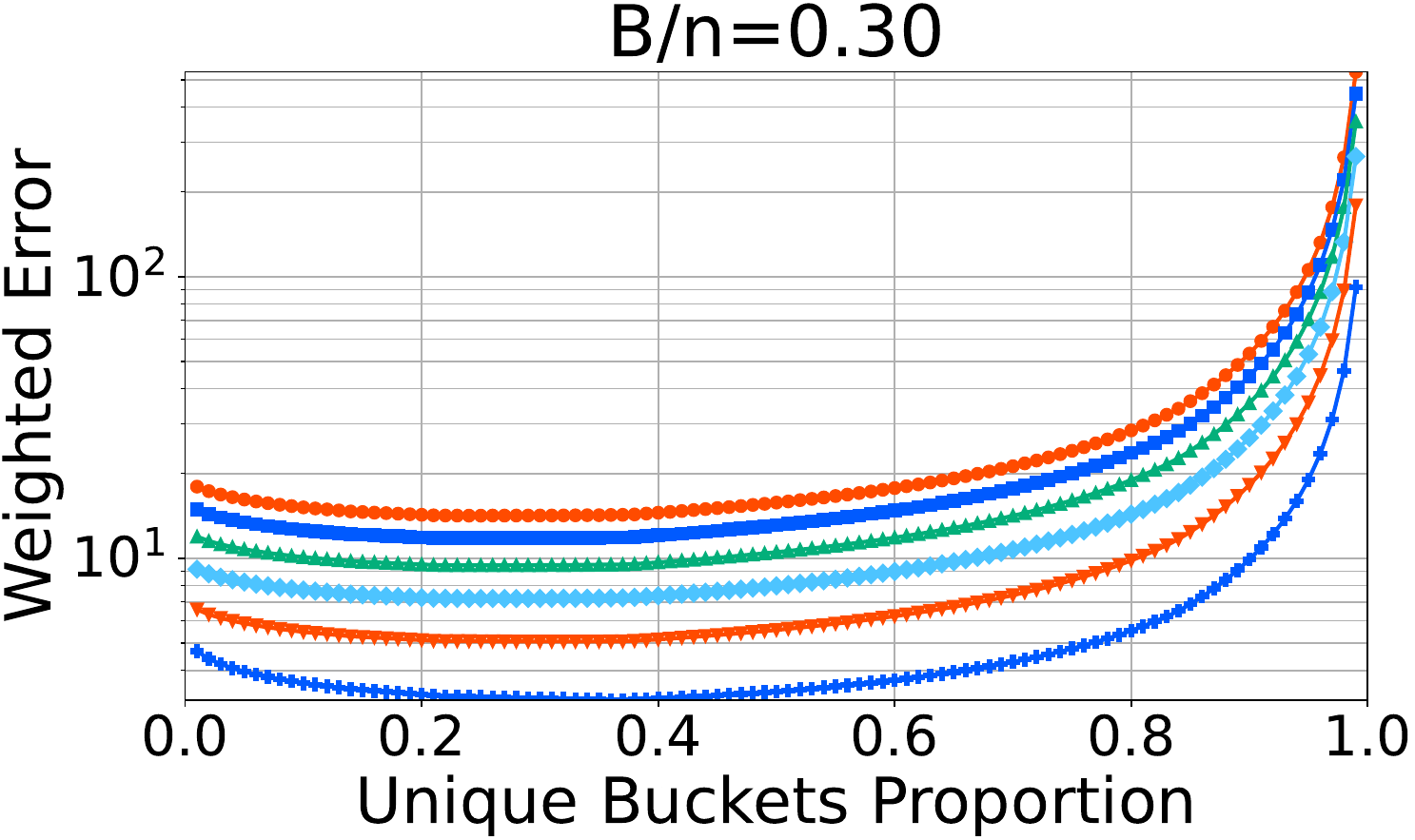}
    \end{minipage}
    \begin{minipage}[tb]{0.48\columnwidth}
        \centering
        \includegraphics[width=\columnwidth]{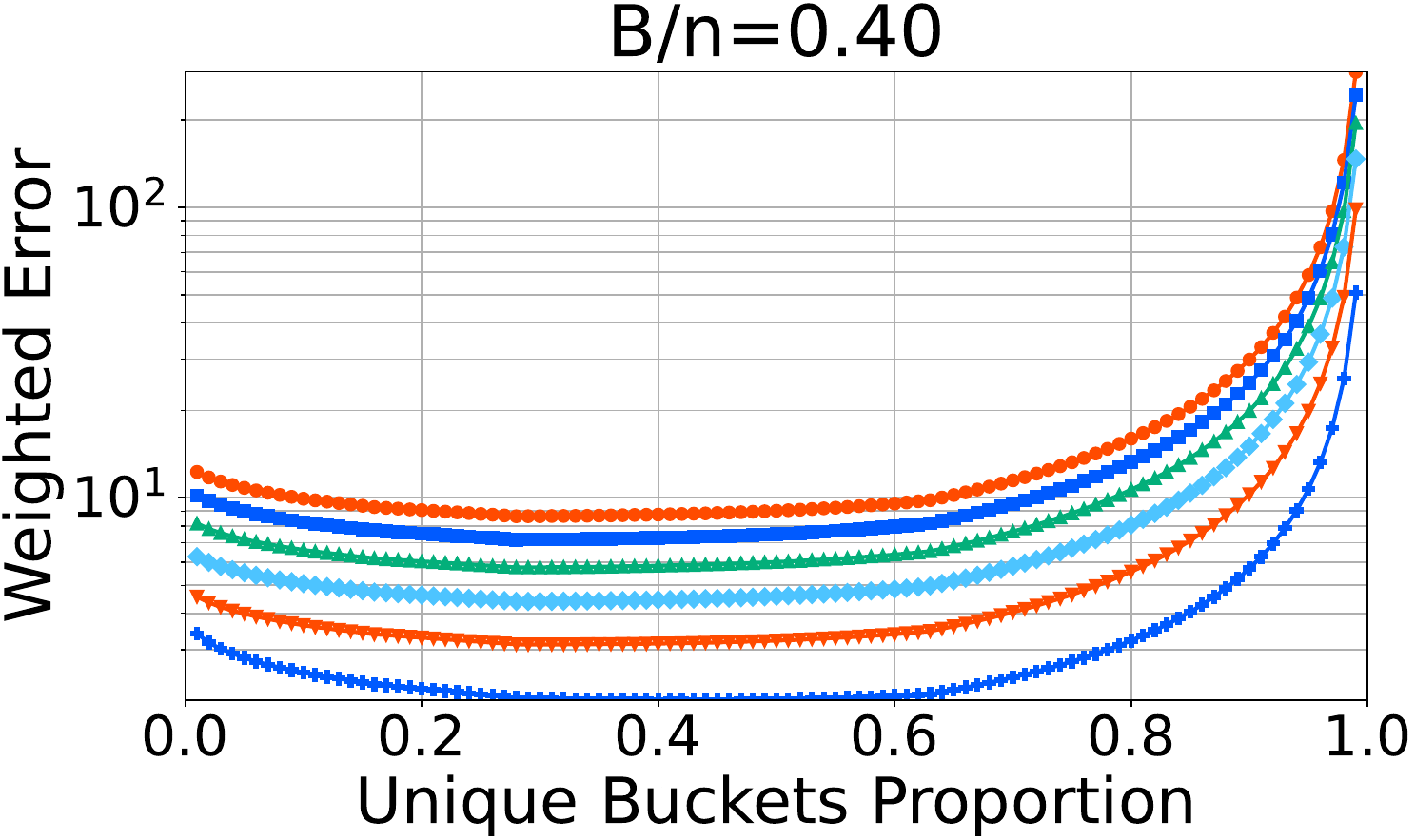}
    \end{minipage}
    \begin{minipage}[tb]{0.48\columnwidth}
        \centering
        \includegraphics[width=\columnwidth]{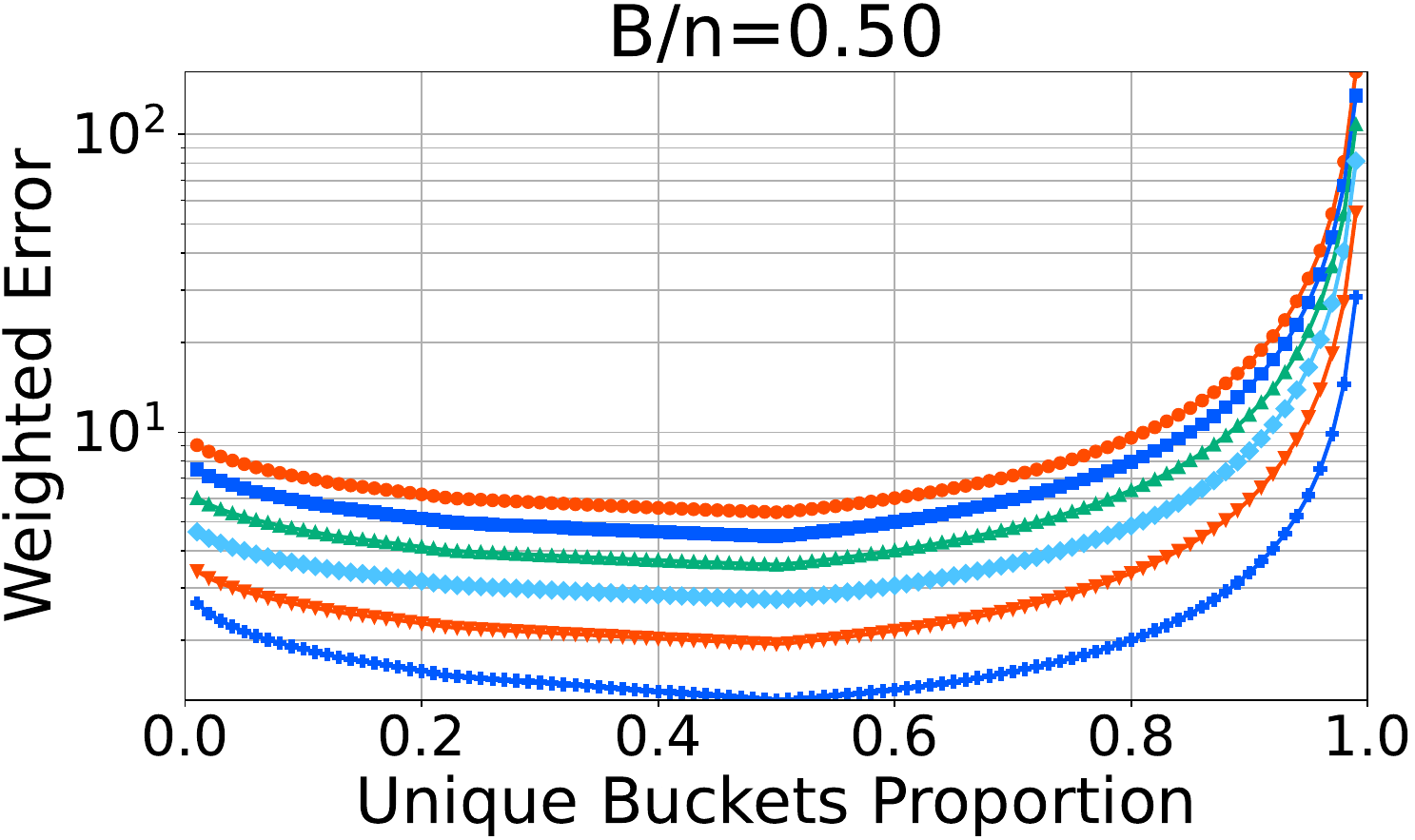}
    \end{minipage}
    \begin{minipage}[tb]{0.48\columnwidth}
        \centering
        \includegraphics[width=\columnwidth]{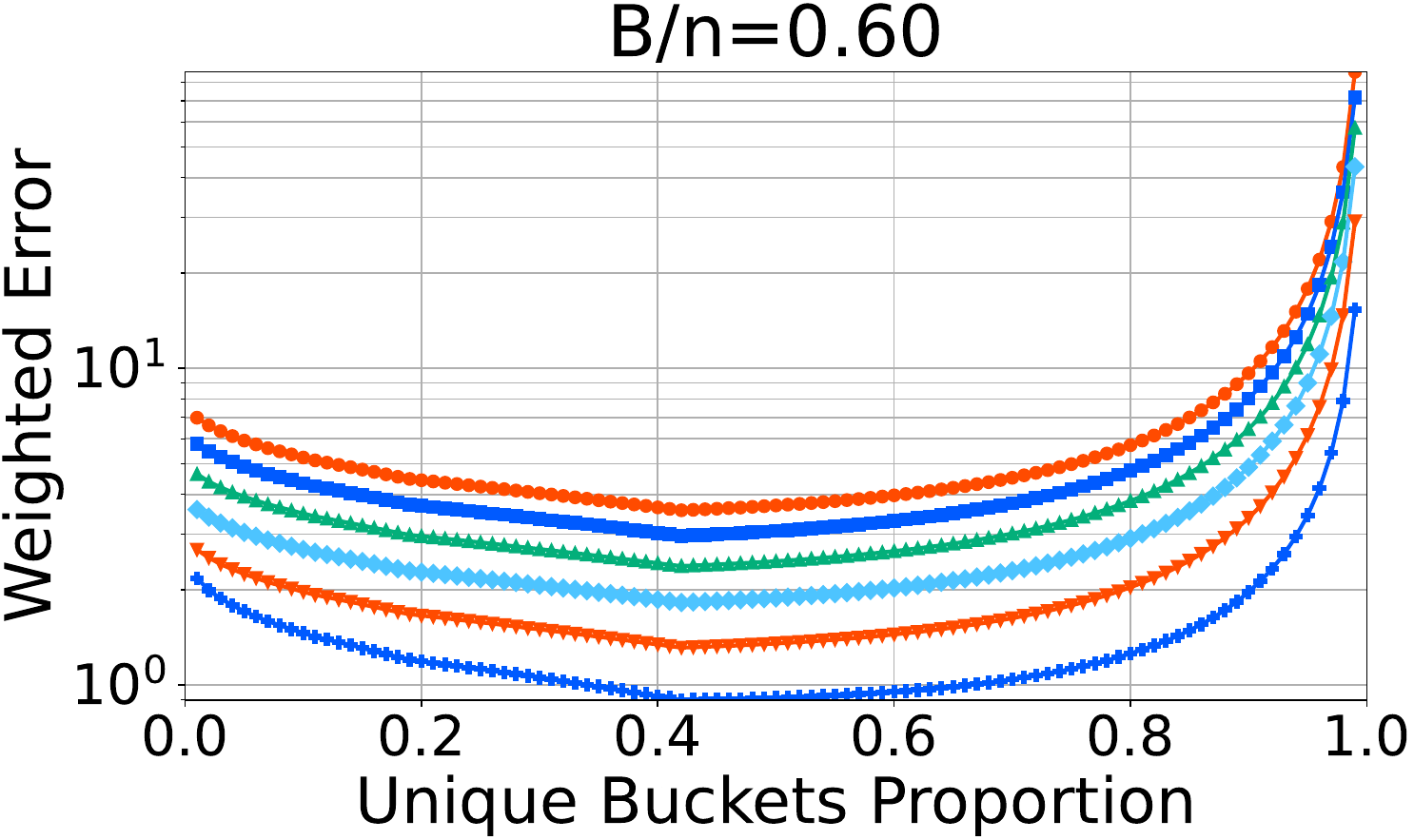}
    \end{minipage}
    \begin{minipage}[tb]{0.48\columnwidth}
        \centering
        \includegraphics[width=\columnwidth]{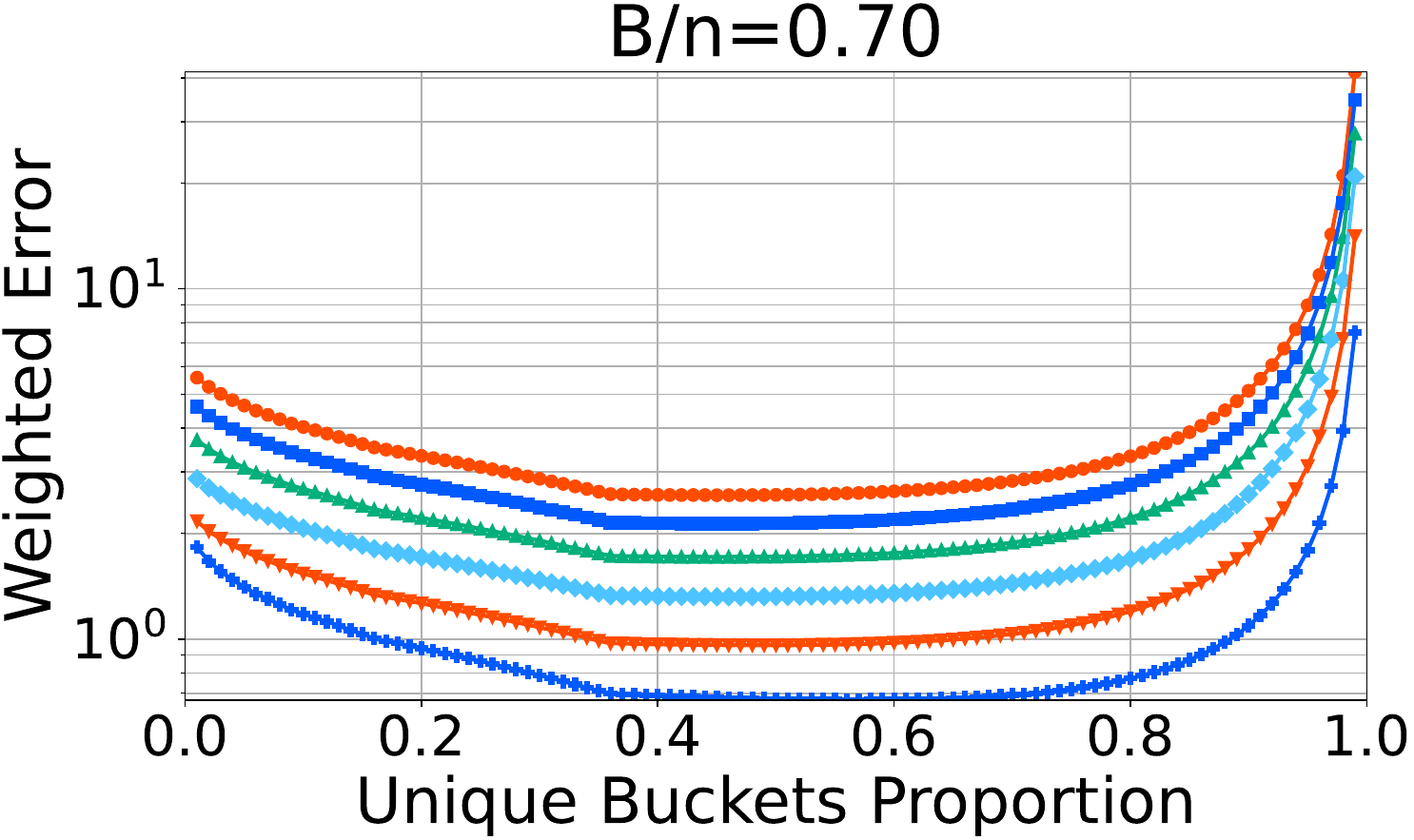}
    \end{minipage}
    \begin{minipage}[tb]{0.48\columnwidth}
        \centering
        \includegraphics[width=\columnwidth]{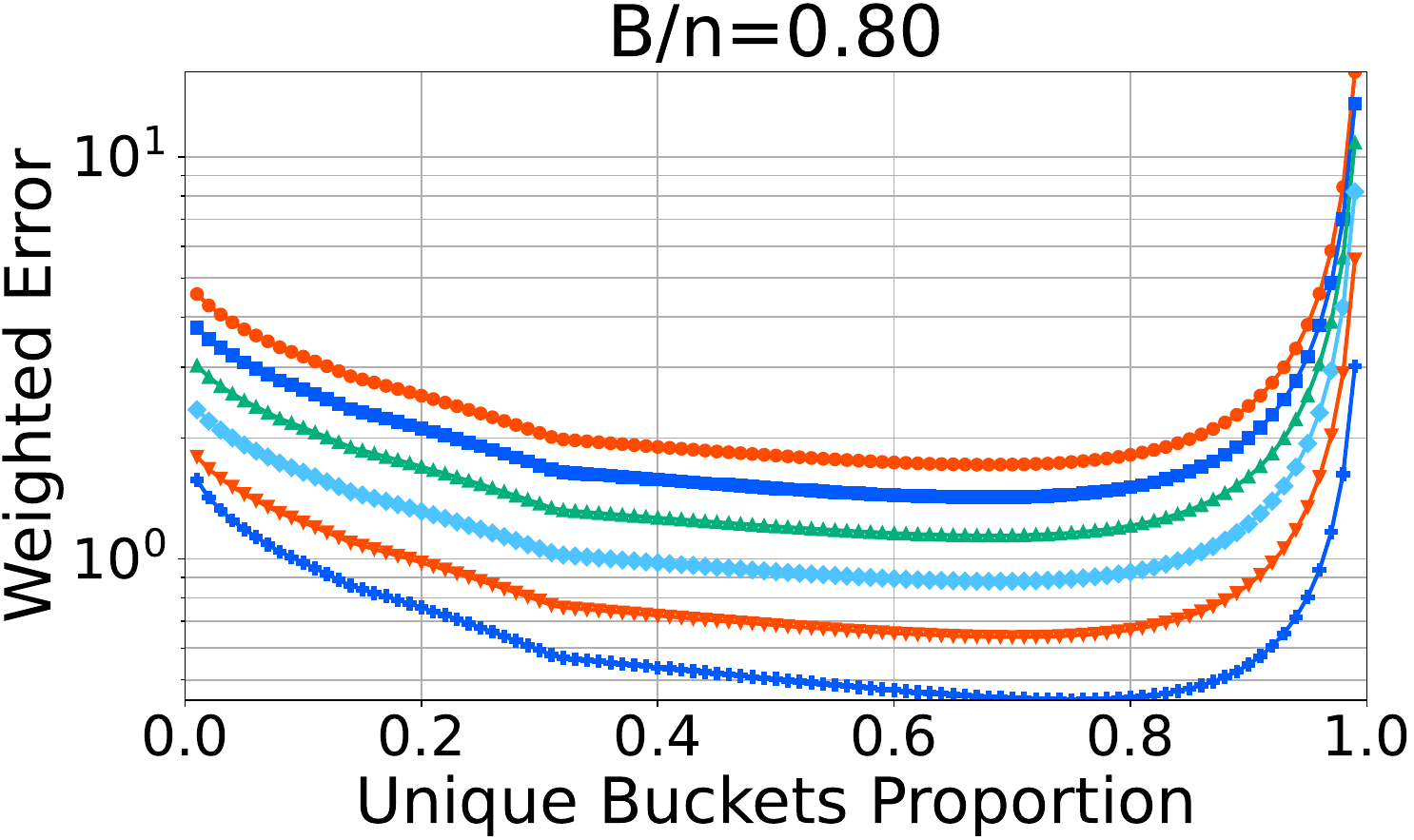}
    \end{minipage}
    \begin{minipage}[tb]{0.48\columnwidth}
        \centering
        \includegraphics[width=\columnwidth]{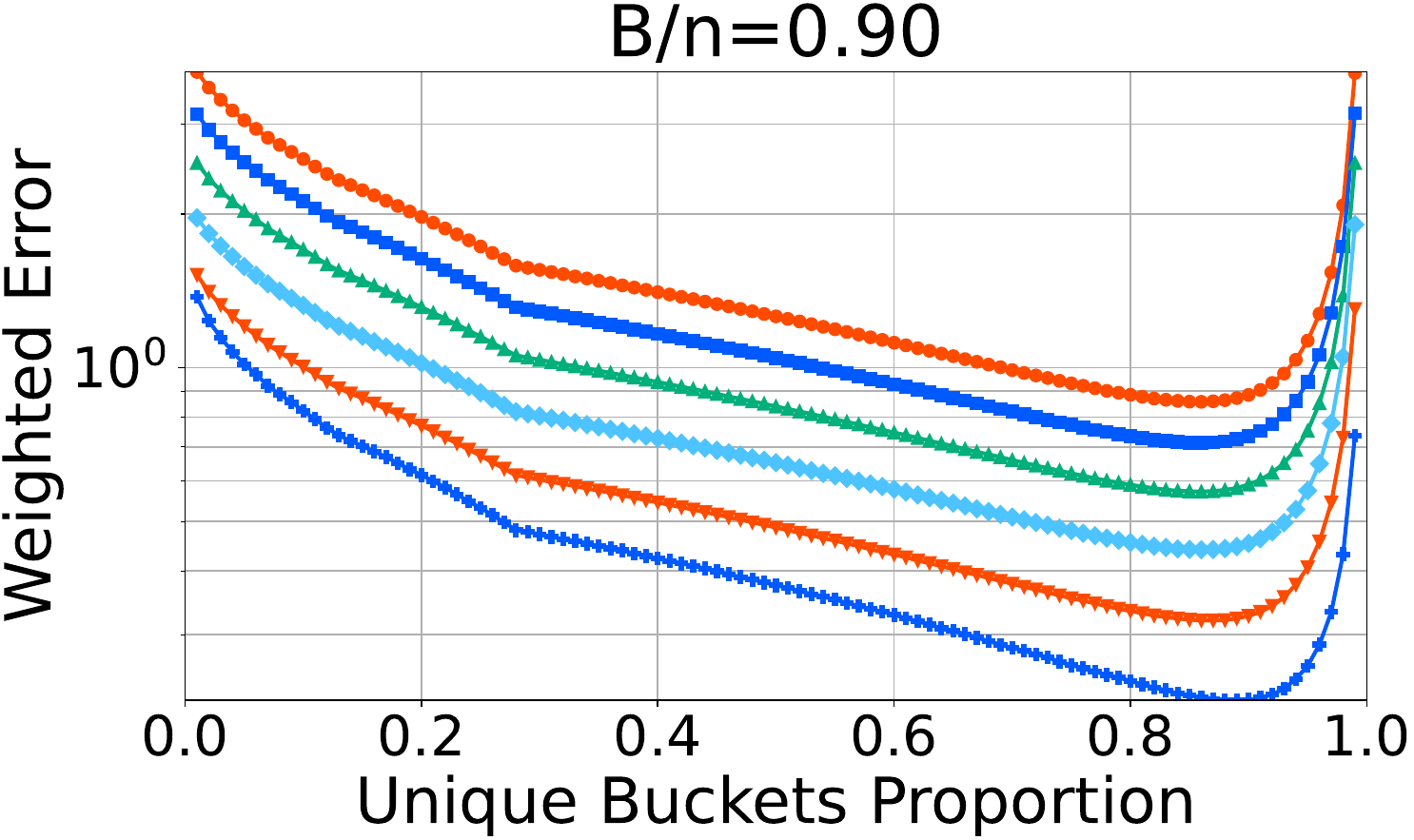}
    \end{minipage}
    \begin{minipage}[tb]{0.15\columnwidth}
        \centering
        \includegraphics[width=\columnwidth]{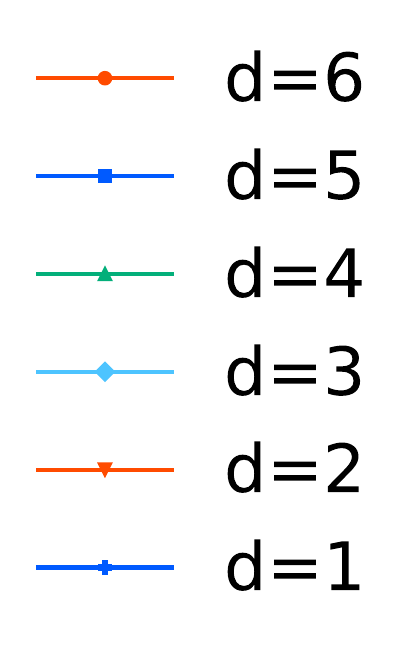}
    \end{minipage}

    \caption{The relationship between $2$ types of parameters and the weighted error on AOL Query Logs Dataset (Average over $10$ runs)}
    \label{fig:20}
\end{figure}

\subsection{Wikipedia Pageviews Dataset}
The experimental results obtained from the Wikipedia Pageviews Dataset \cite{wikimedia2026pageviews} are shown in \cref{fig:21}.
\begin{figure}[tb]
    \centering
    \begin{minipage}[tb]{0.48\columnwidth}
        \centering
        \includegraphics[width=\columnwidth]{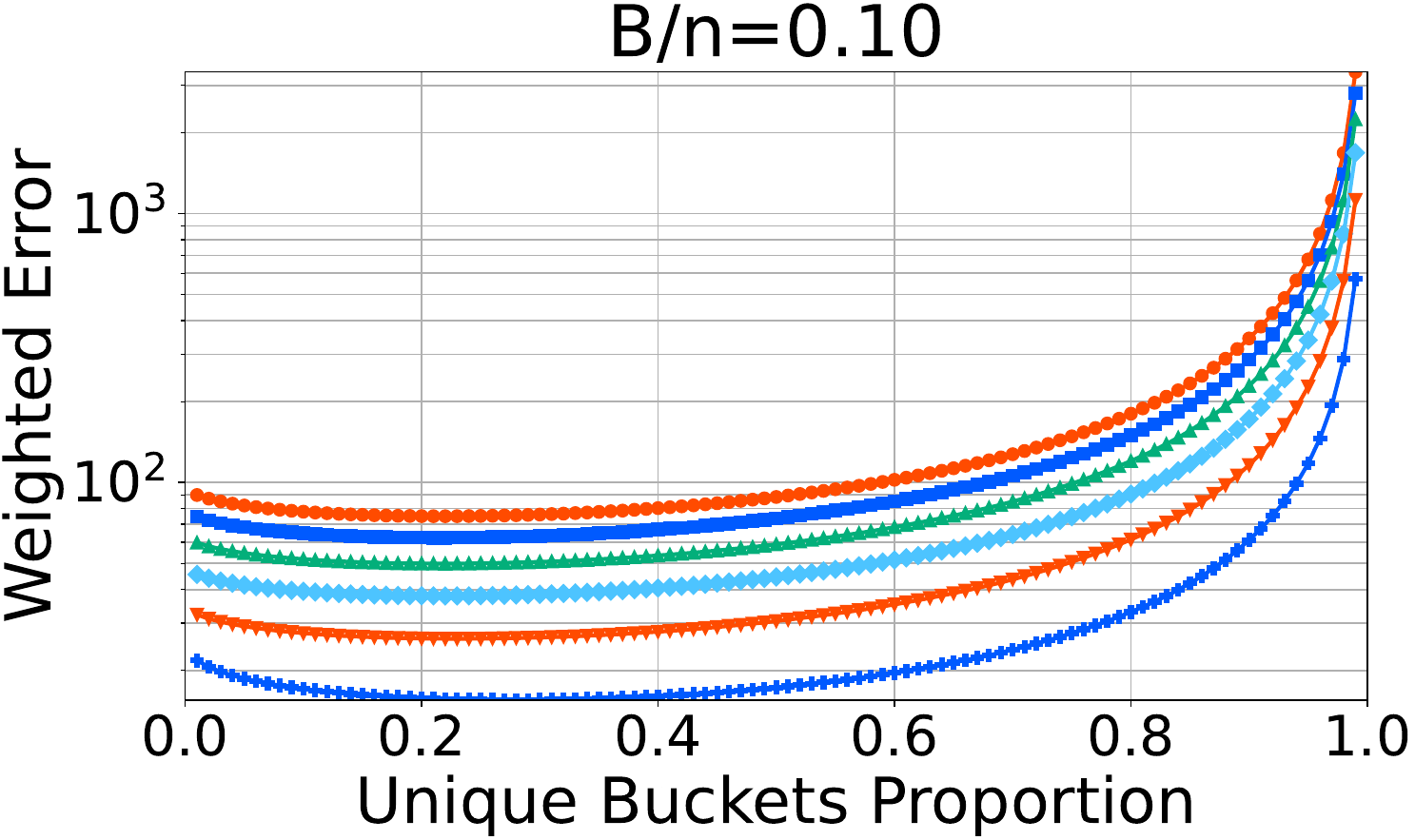}
    \end{minipage}
    \begin{minipage}[tb]{0.48\columnwidth}
        \centering
        \includegraphics[width=\columnwidth]{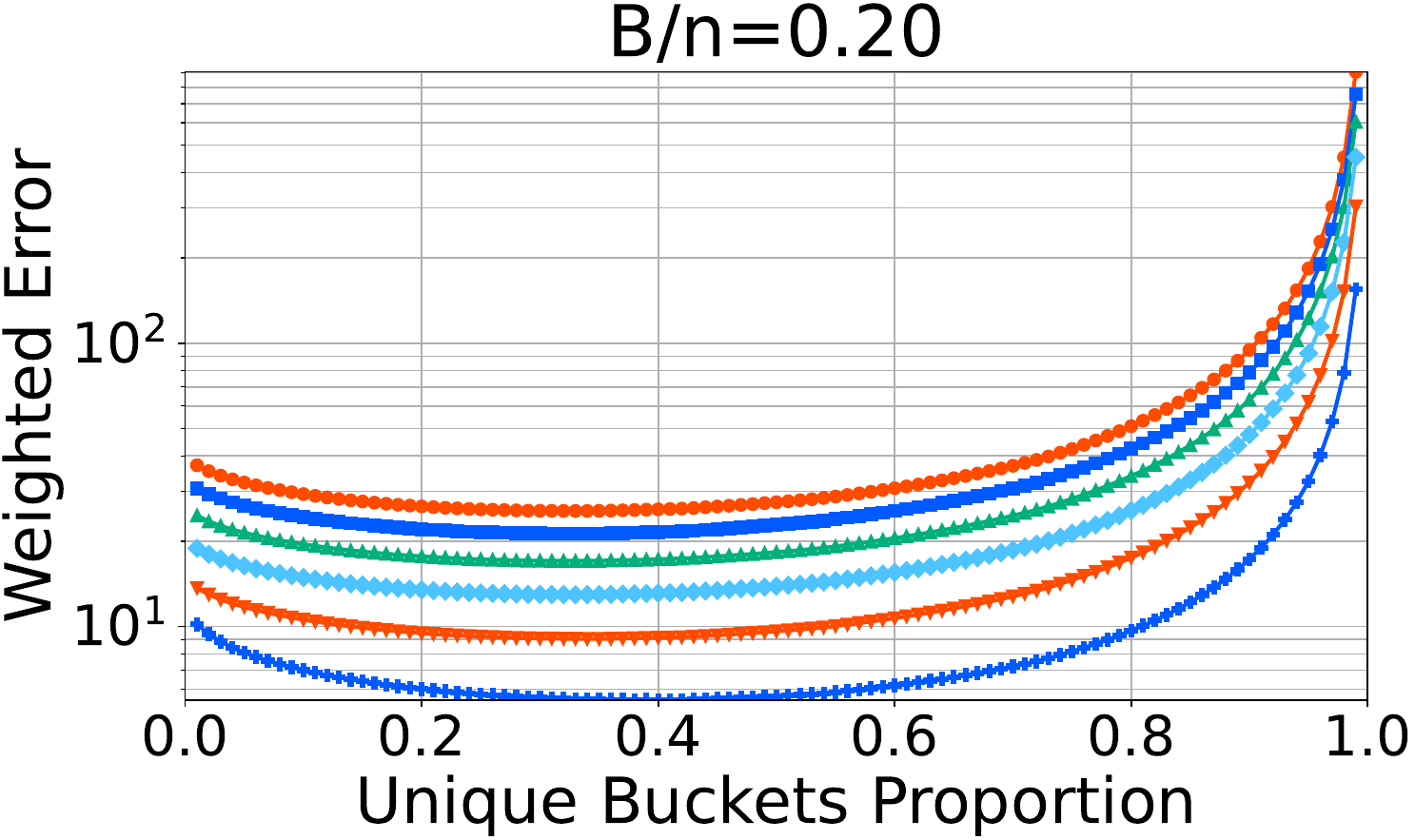}
    \end{minipage}
    \begin{minipage}[tb]{0.48\columnwidth}
        \centering
        \includegraphics[width=\columnwidth]{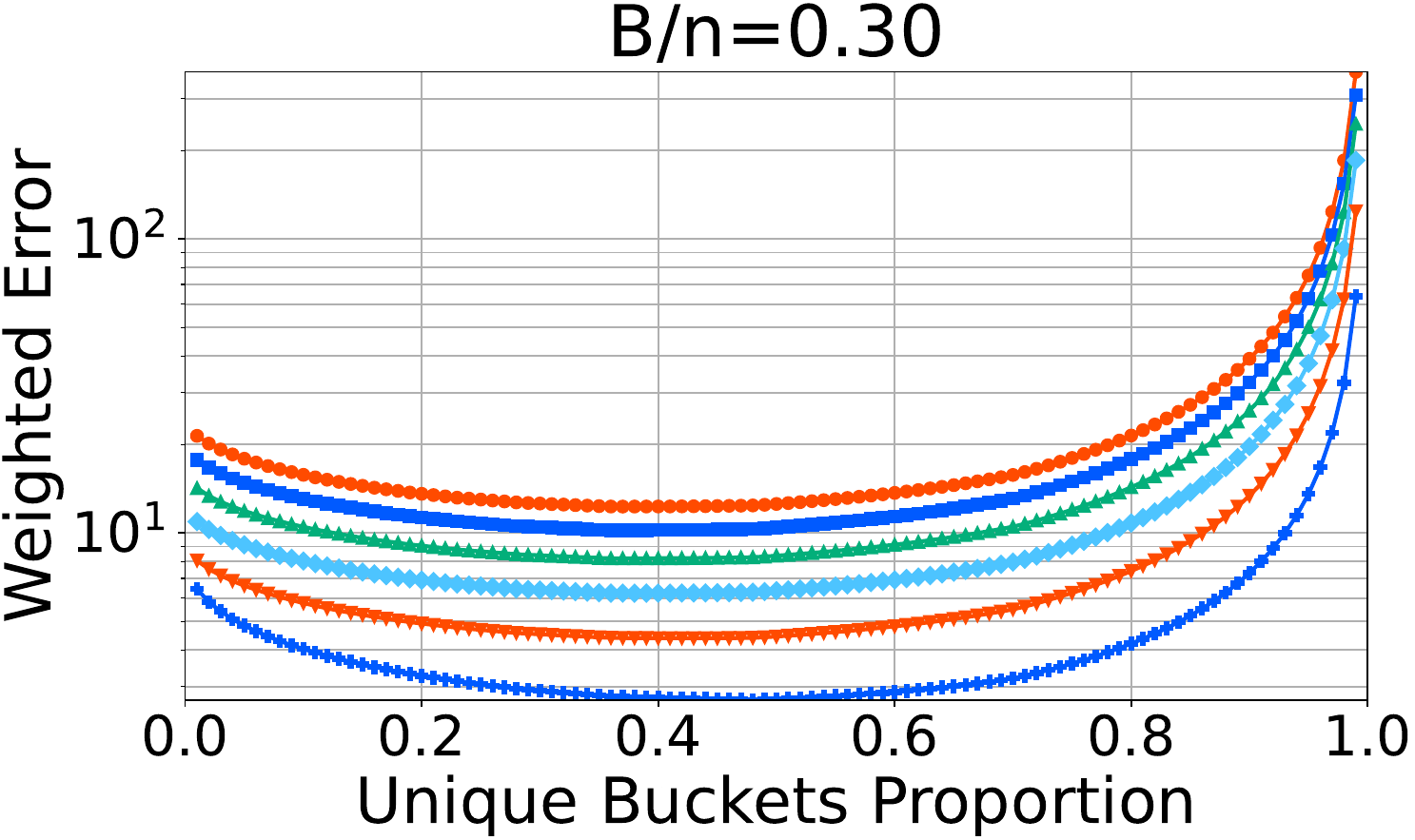}
    \end{minipage}
    \begin{minipage}[tb]{0.48\columnwidth}
        \centering
        \includegraphics[width=\columnwidth]{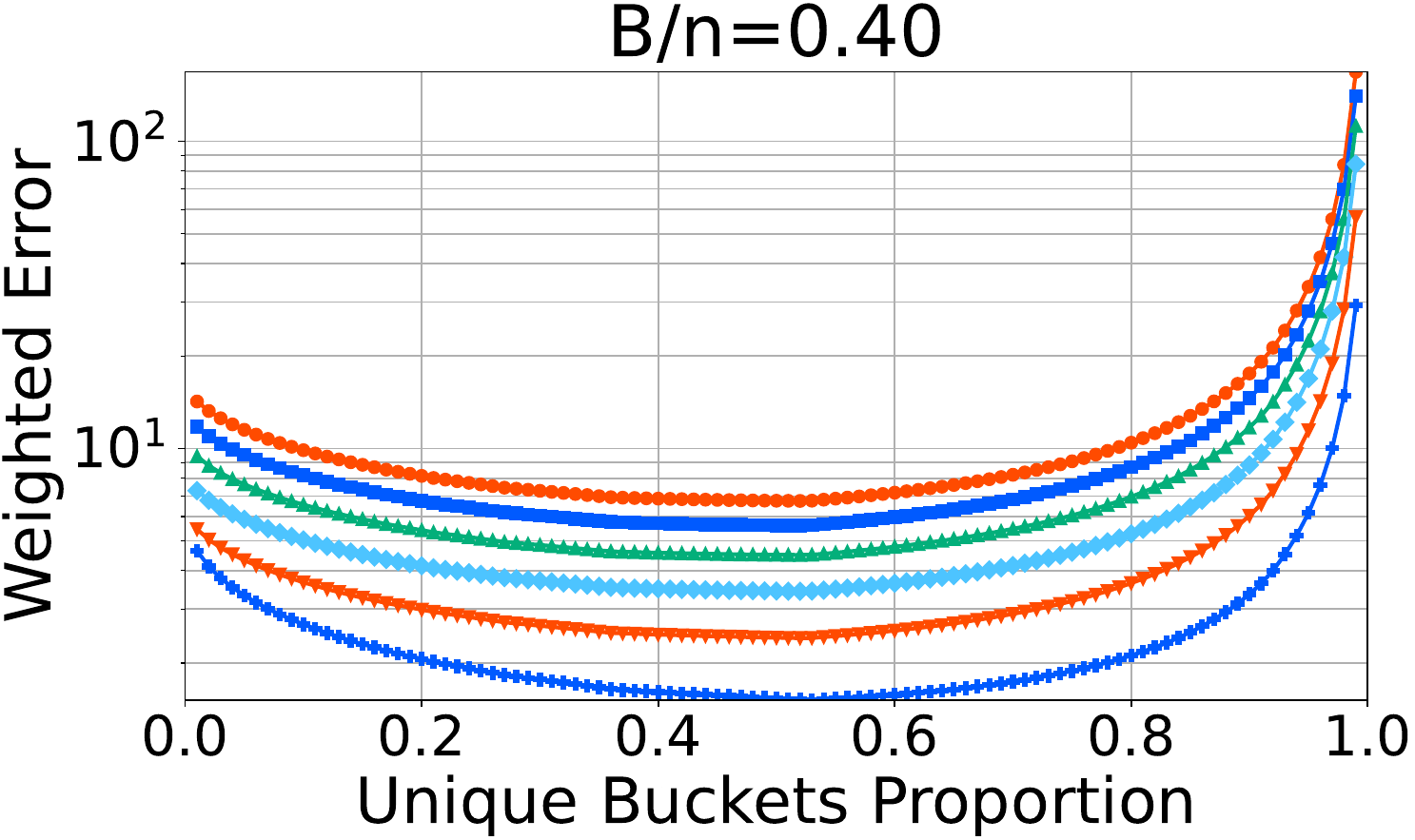}
    \end{minipage}
    \begin{minipage}[tb]{0.48\columnwidth}
        \centering
        \includegraphics[width=\columnwidth]{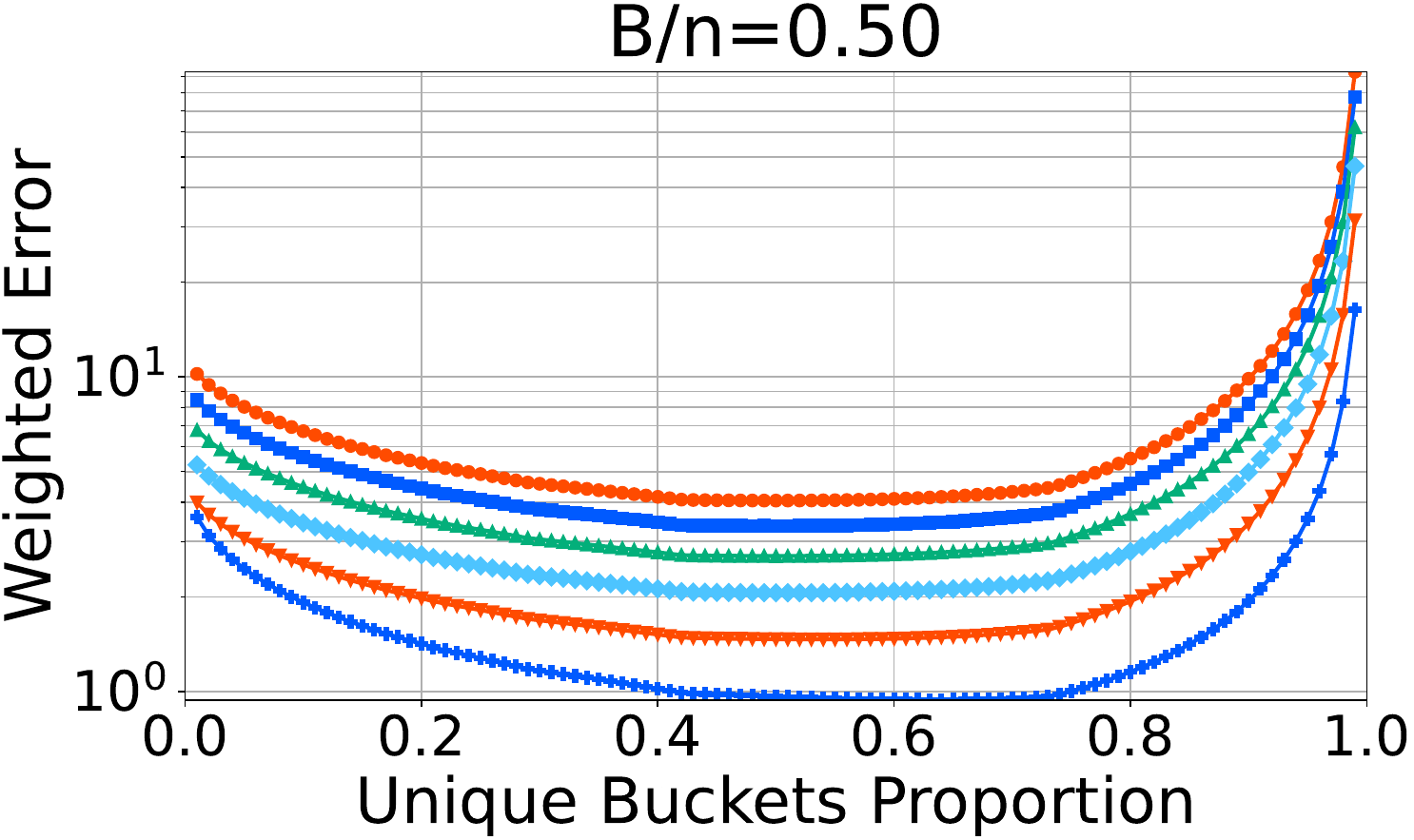}
    \end{minipage}
    \begin{minipage}[tb]{0.48\columnwidth}
        \centering
        \includegraphics[width=\columnwidth]{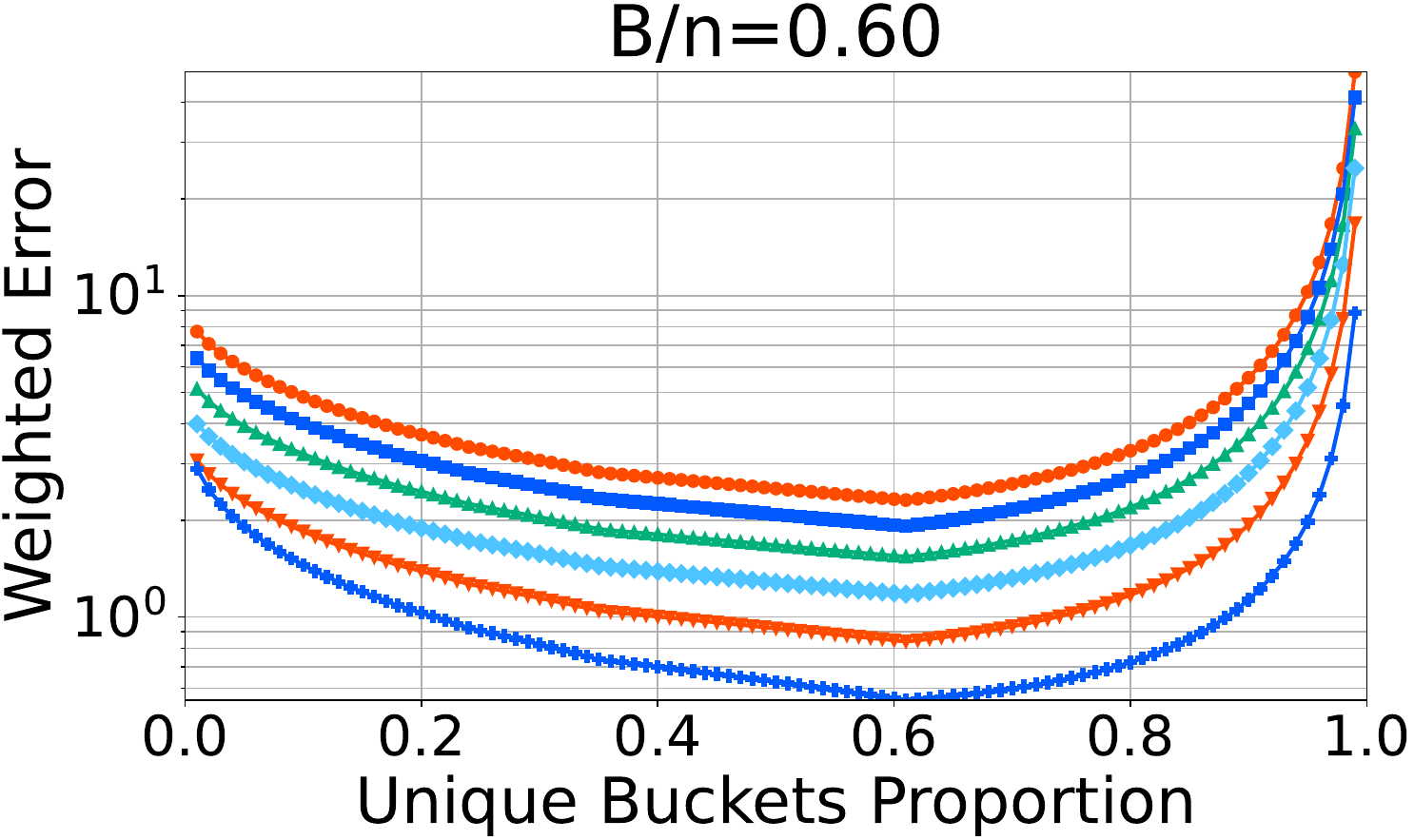}
    \end{minipage}
    \begin{minipage}[tb]{0.48\columnwidth}
        \centering
        \includegraphics[width=\columnwidth]{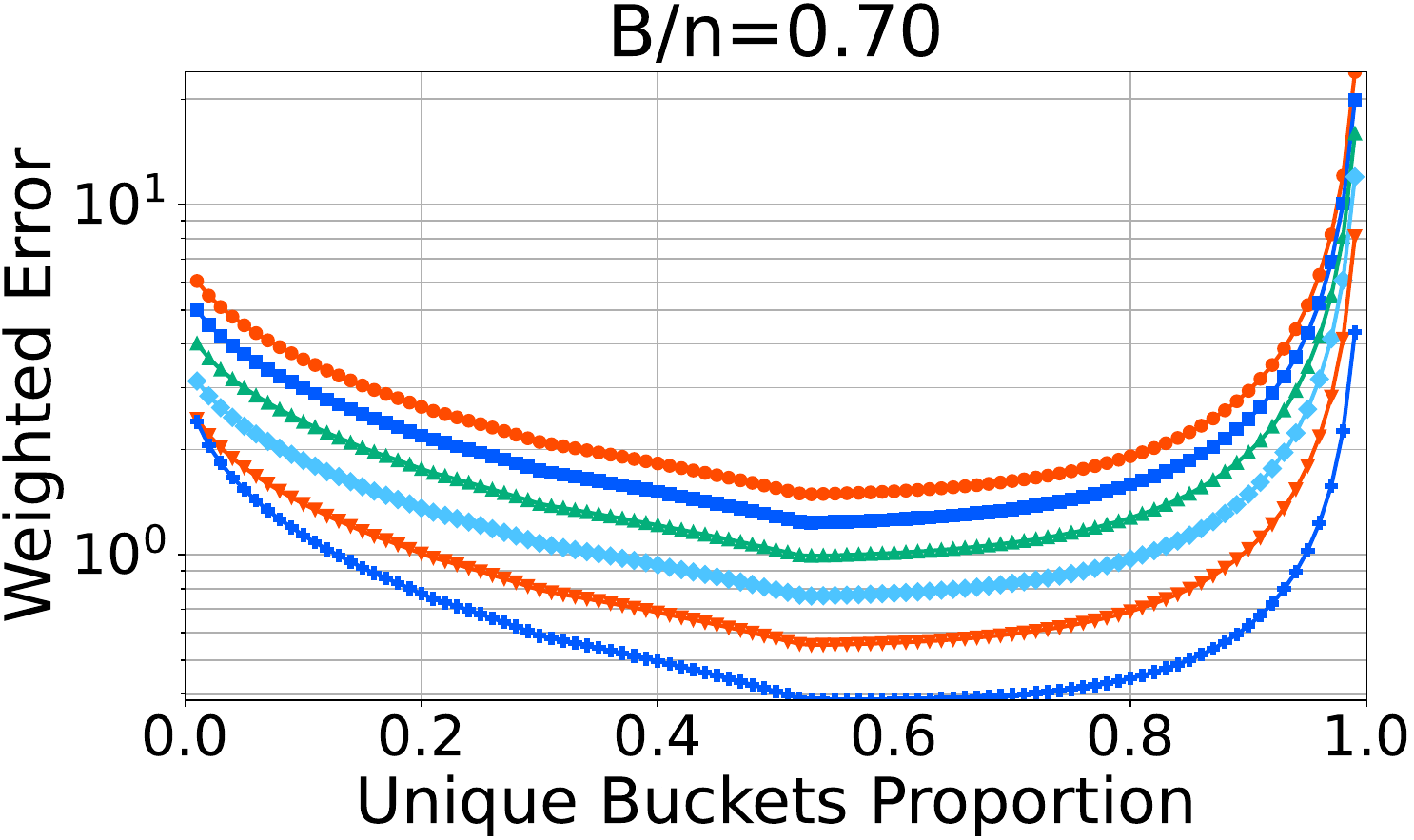}
    \end{minipage}
    \begin{minipage}[tb]{0.48\columnwidth}
        \centering
        \includegraphics[width=\columnwidth]{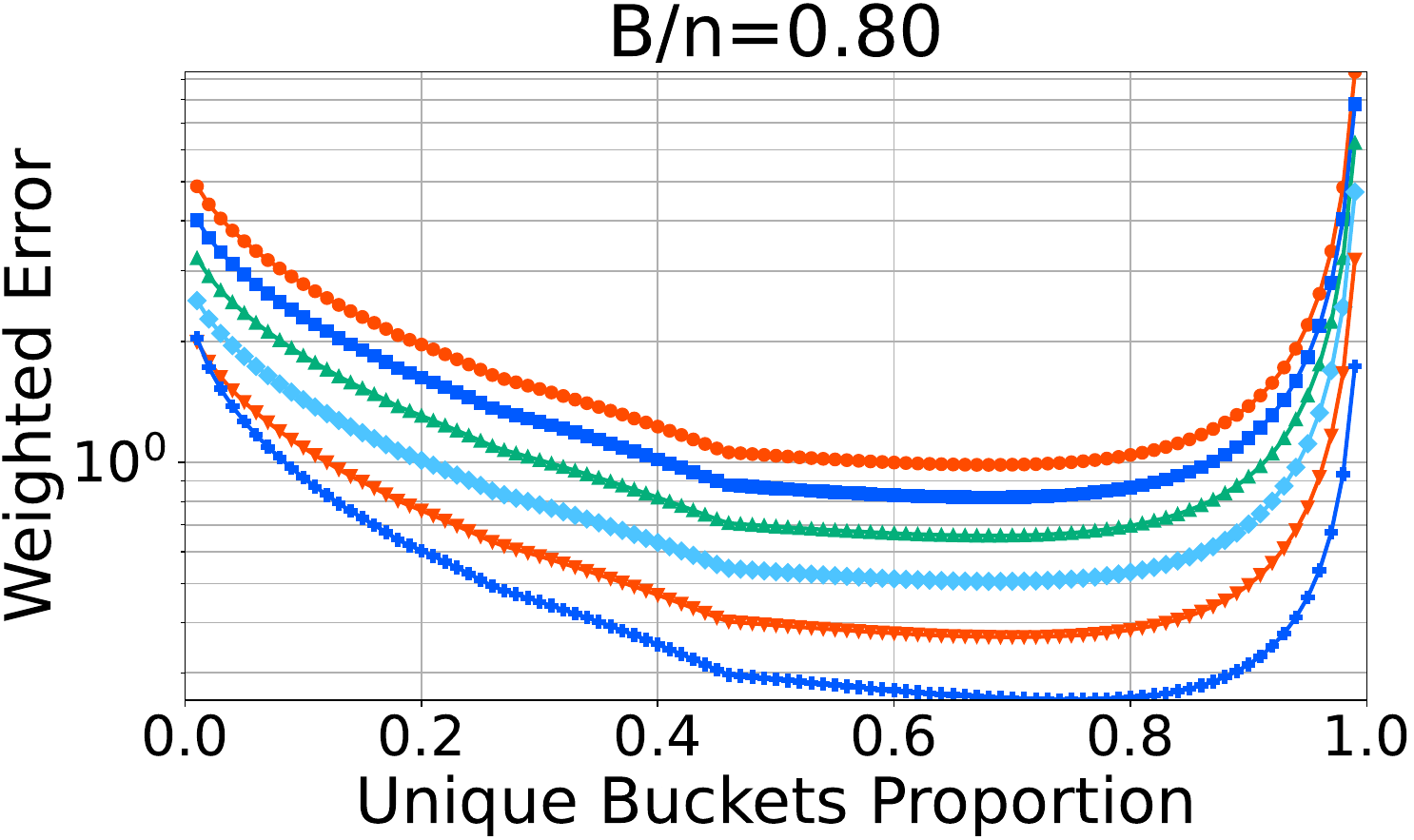}
    \end{minipage}
    \begin{minipage}[tb]{0.48\columnwidth}
        \centering
        \includegraphics[width=\columnwidth]{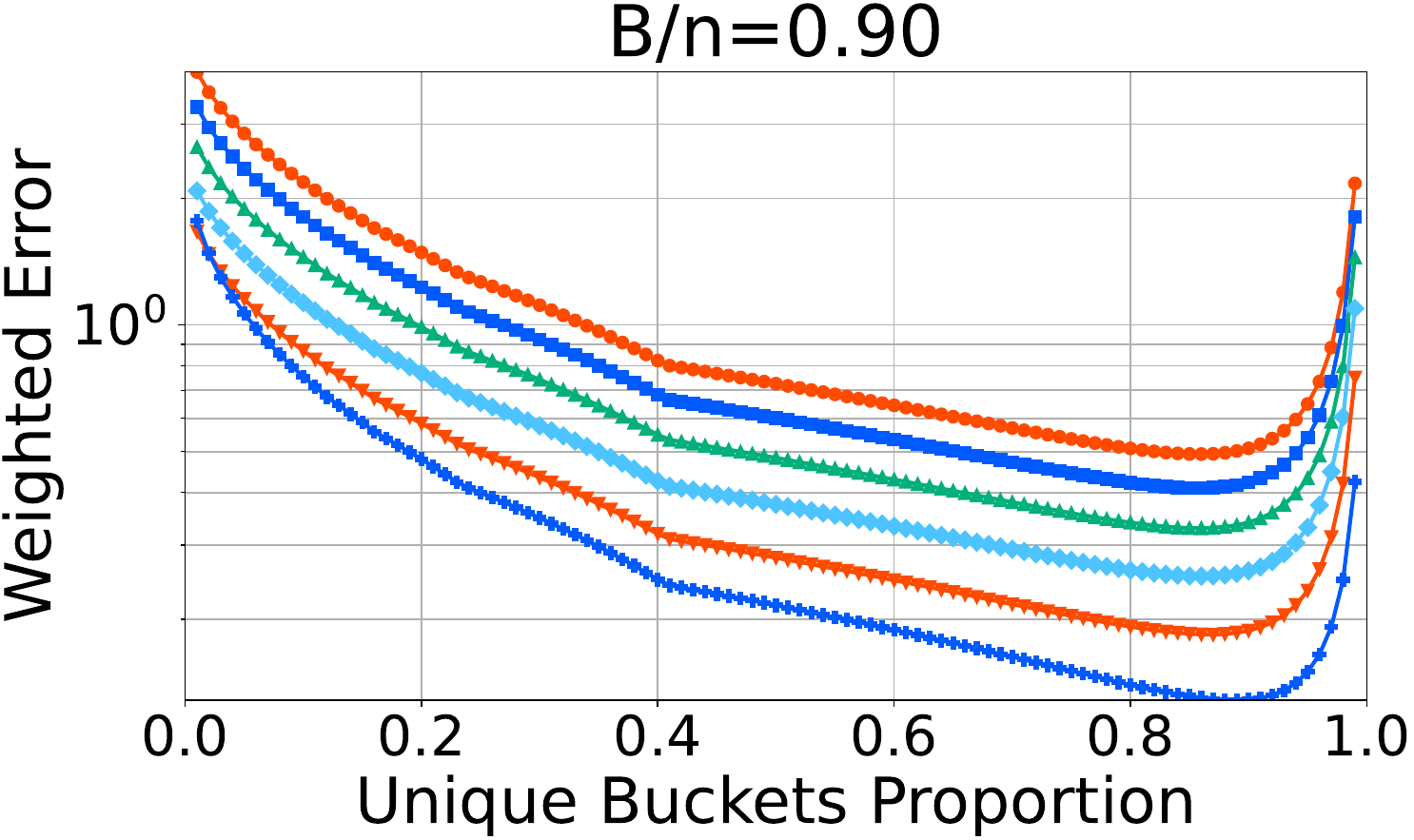}
    \end{minipage}
    \begin{minipage}[tb]{0.15\columnwidth}
        \centering
        \includegraphics[width=\columnwidth]{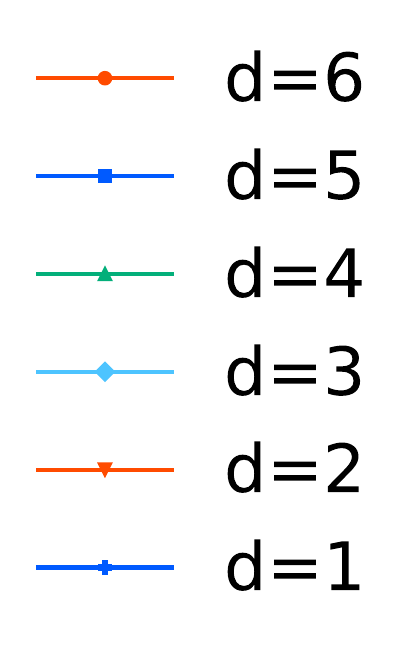}
    \end{minipage}
    \caption{The relationship between $2$ types of parameters and the weighted error on Wikipedia Pageviews Dataset (Average over $10$ runs)}
    \label{fig:21}
\end{figure}

\subsection{Google Books Ngram Viewer Datasets}
The experimental results obtained from the Google Books Ngram Viewer Datasets \cite{michel2011quantitative} are shown in \cref{fig:22}.
\begin{figure}[tb]
    \centering
    \begin{minipage}[tb]{0.48\columnwidth}
        \centering
        \includegraphics[width=\columnwidth]{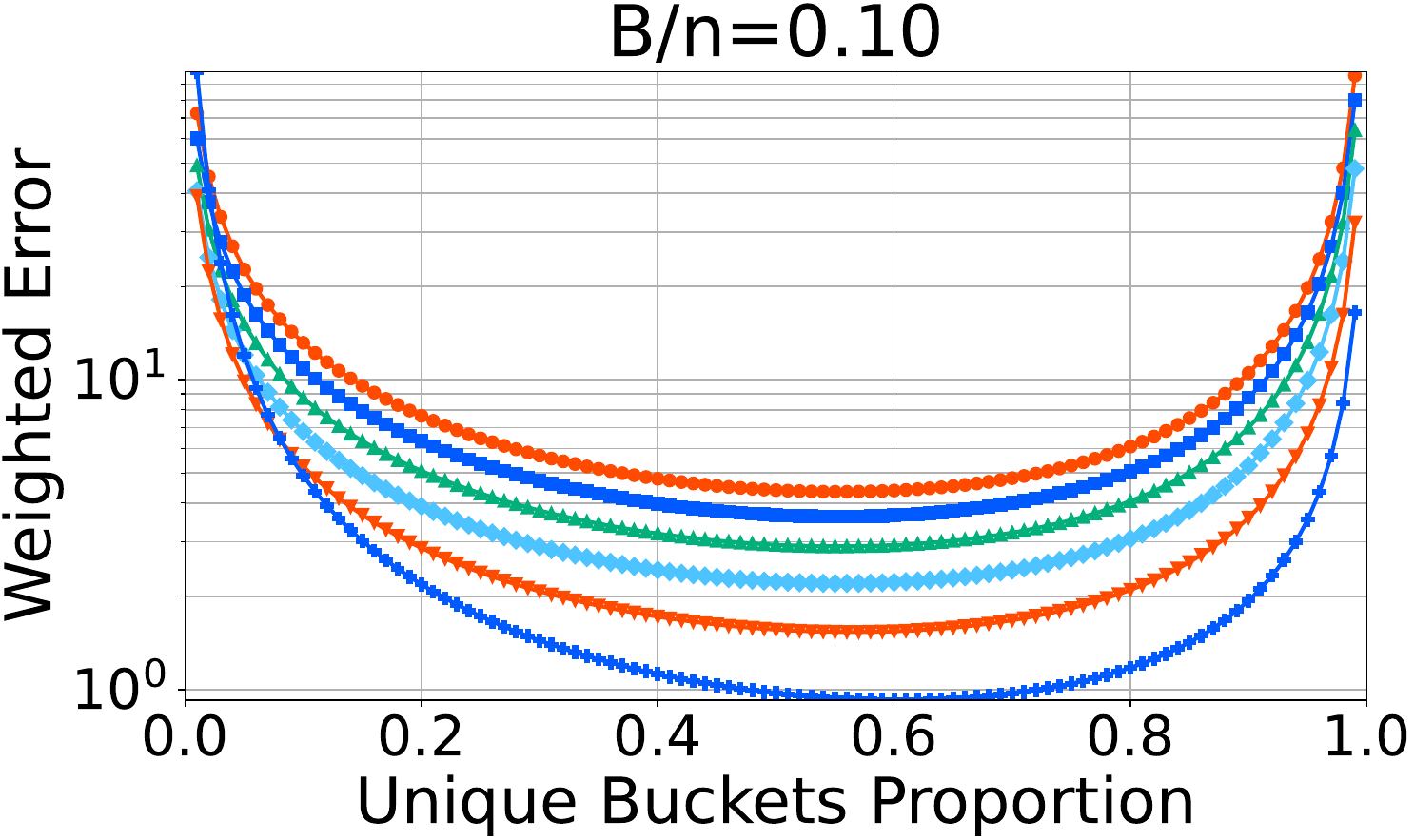}
    \end{minipage}
    \begin{minipage}[tb]{0.48\columnwidth}
        \centering
        \includegraphics[width=\columnwidth]{figures/lcms_error/google/space_0.2000.pdf}
    \end{minipage}
    \begin{minipage}[tb]{0.48\columnwidth}
        \centering
        \includegraphics[width=\columnwidth]{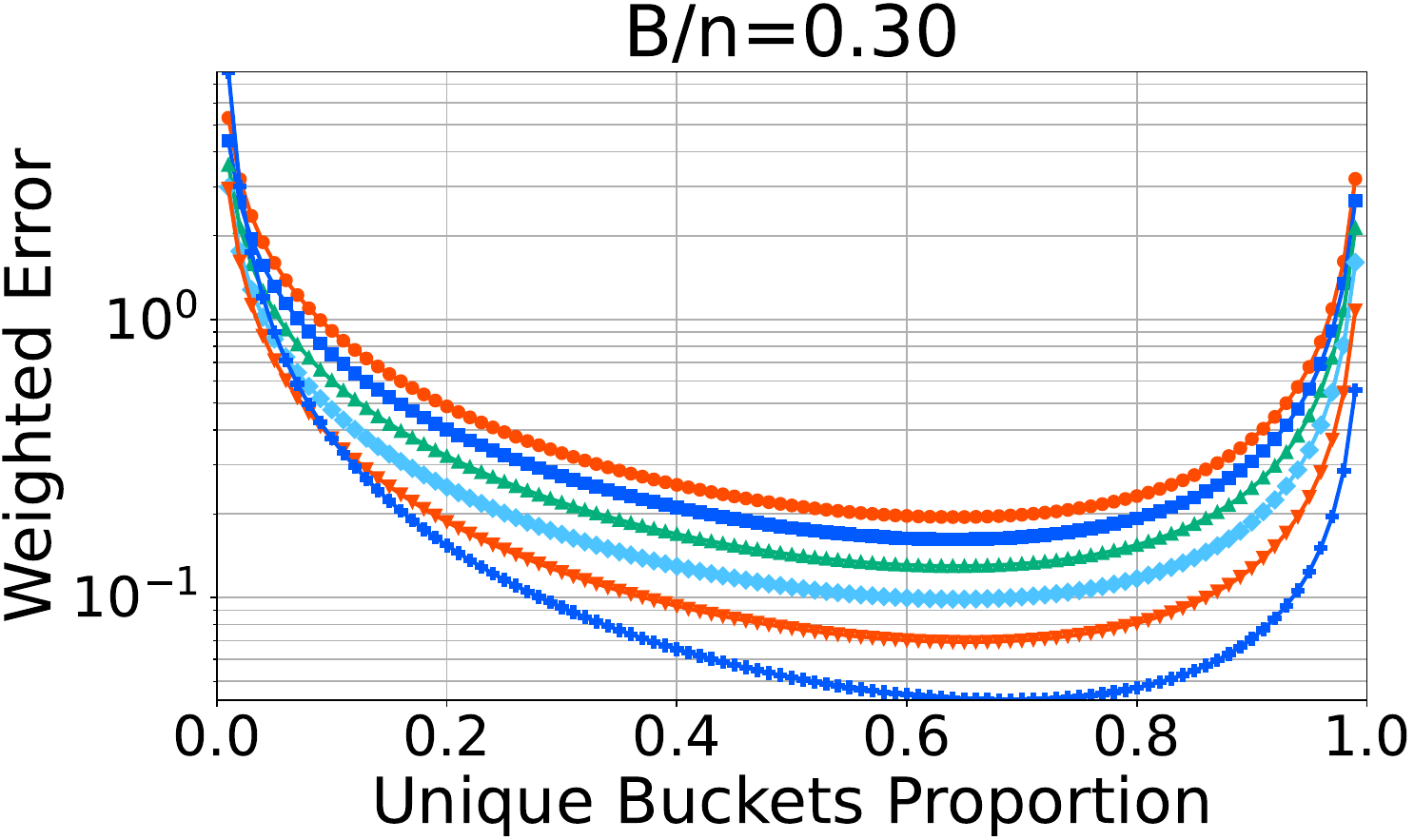}
    \end{minipage}
    \begin{minipage}[tb]{0.48\columnwidth}
        \centering
        \includegraphics[width=\columnwidth]{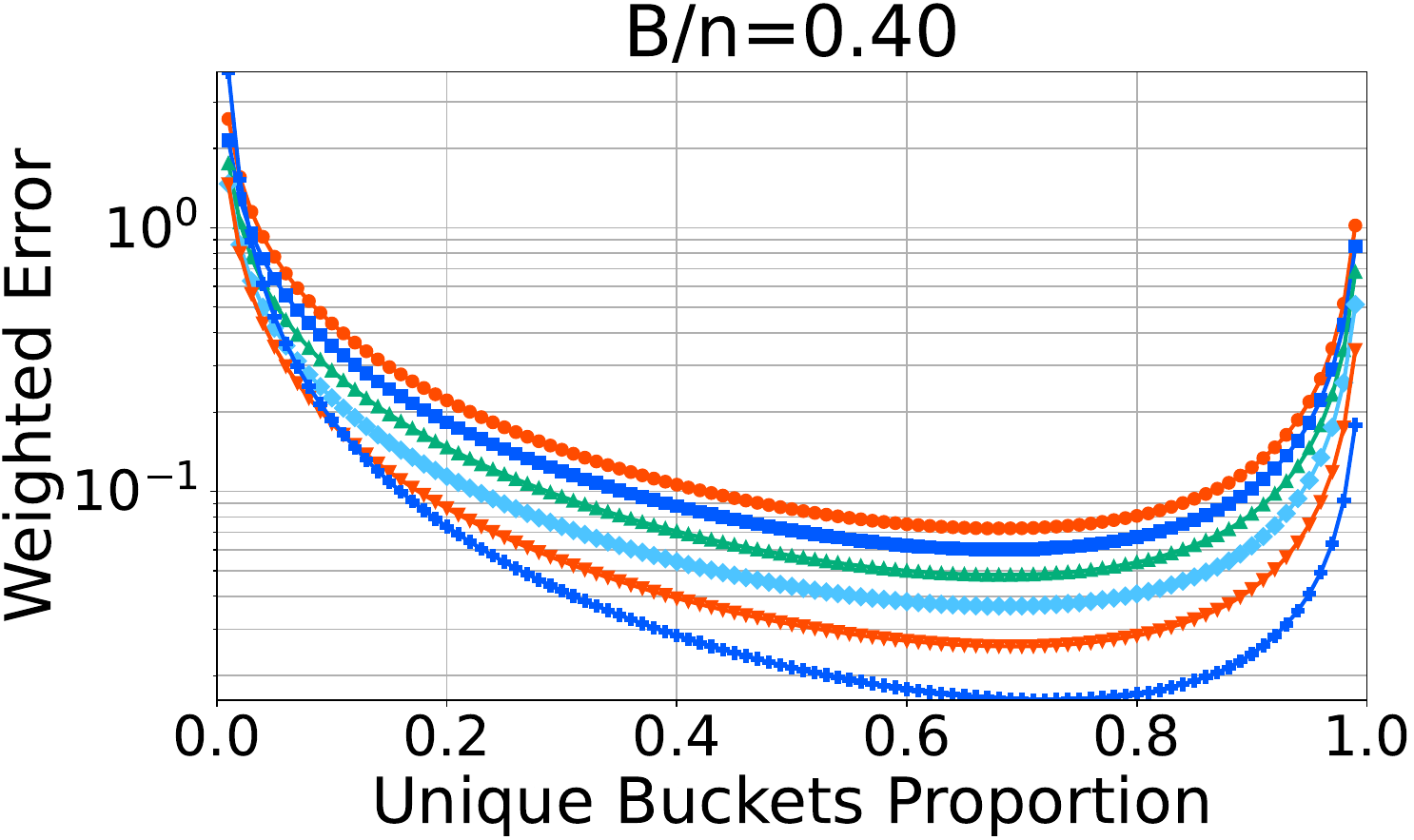}
    \end{minipage}
    \begin{minipage}[tb]{0.48\columnwidth}
        \centering
        \includegraphics[width=\columnwidth]{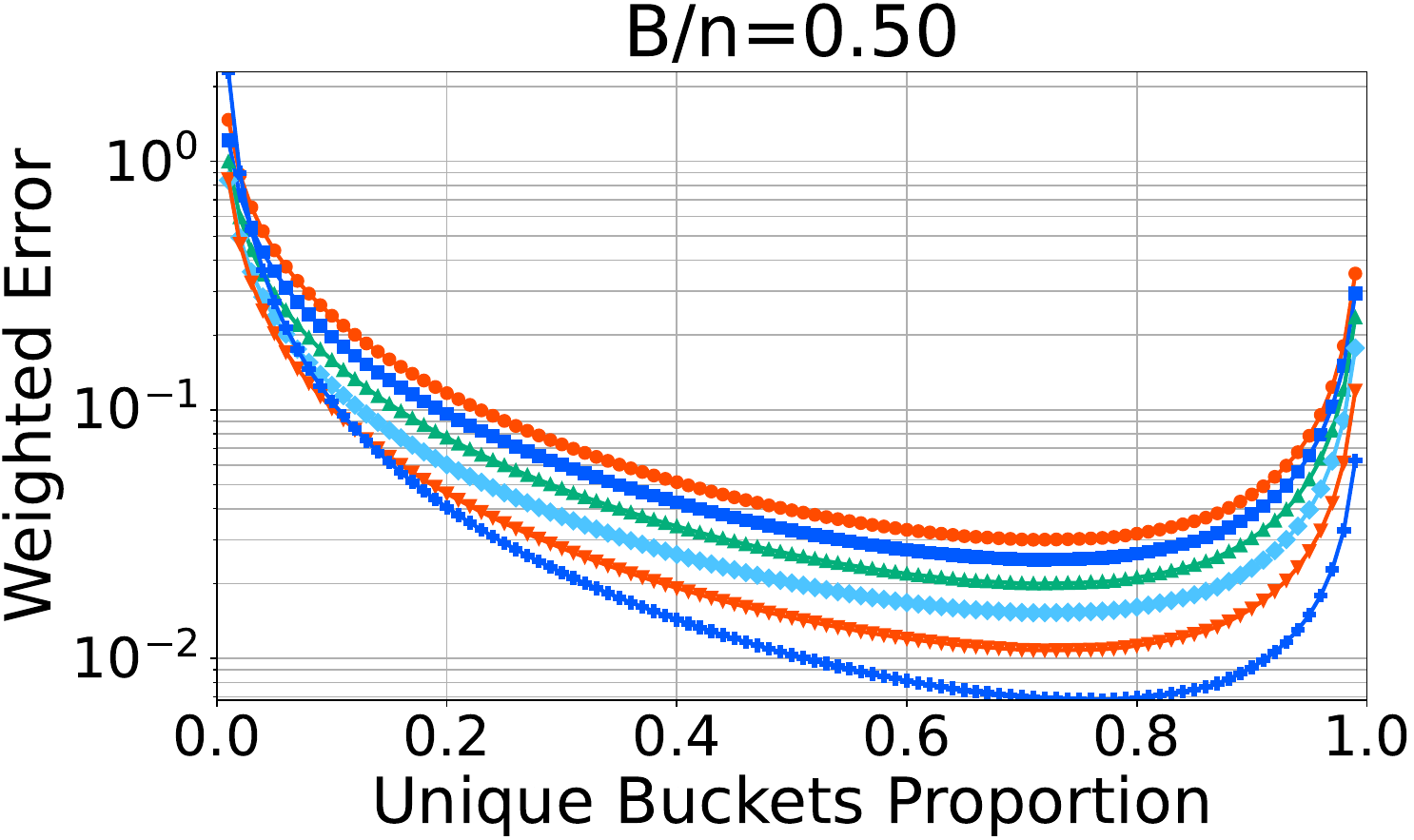}
    \end{minipage}
    \begin{minipage}[tb]{0.48\columnwidth}
        \centering
        \includegraphics[width=\columnwidth]{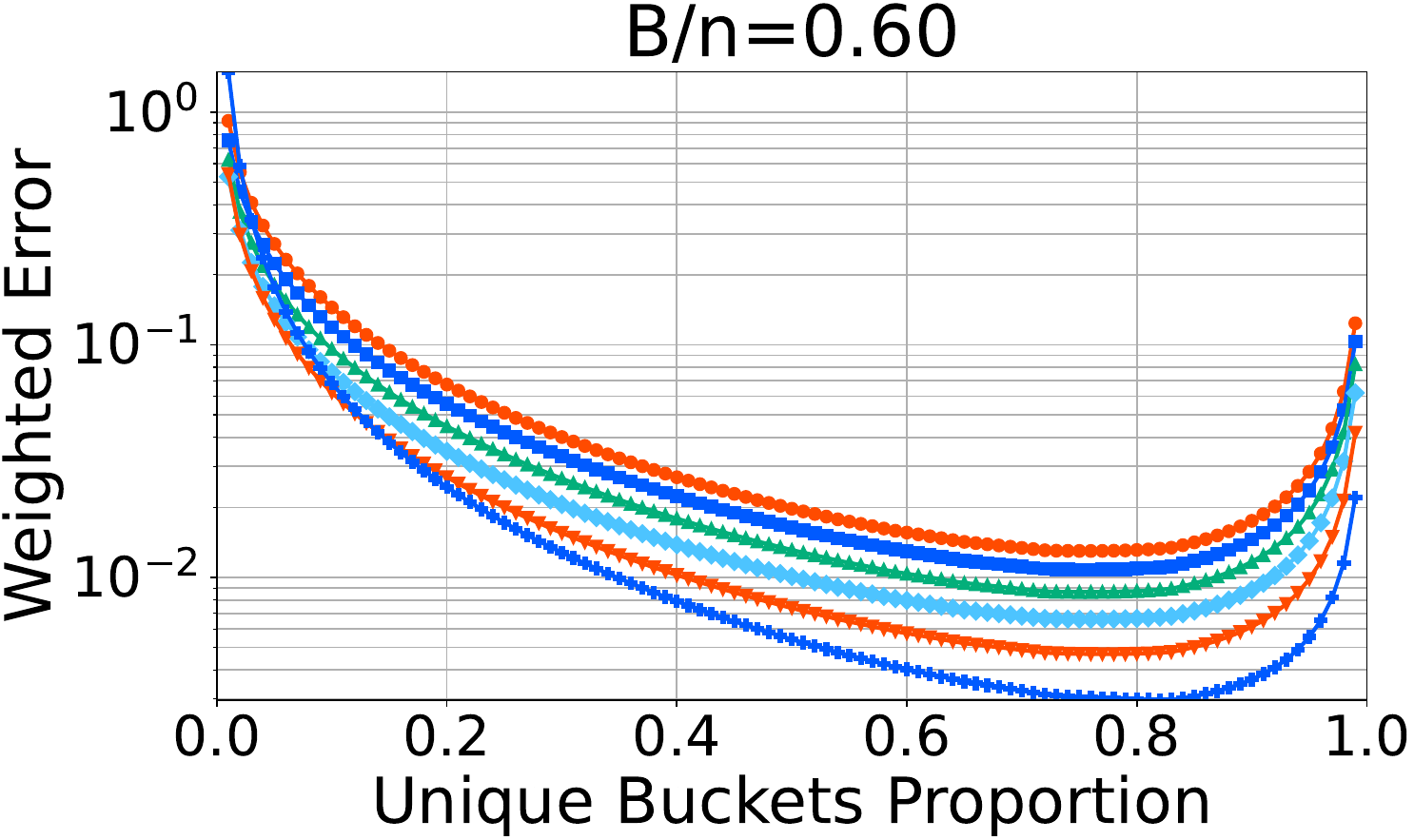}
    \end{minipage}
    \begin{minipage}[tb]{0.48\columnwidth}
        \centering
        \includegraphics[width=\columnwidth]{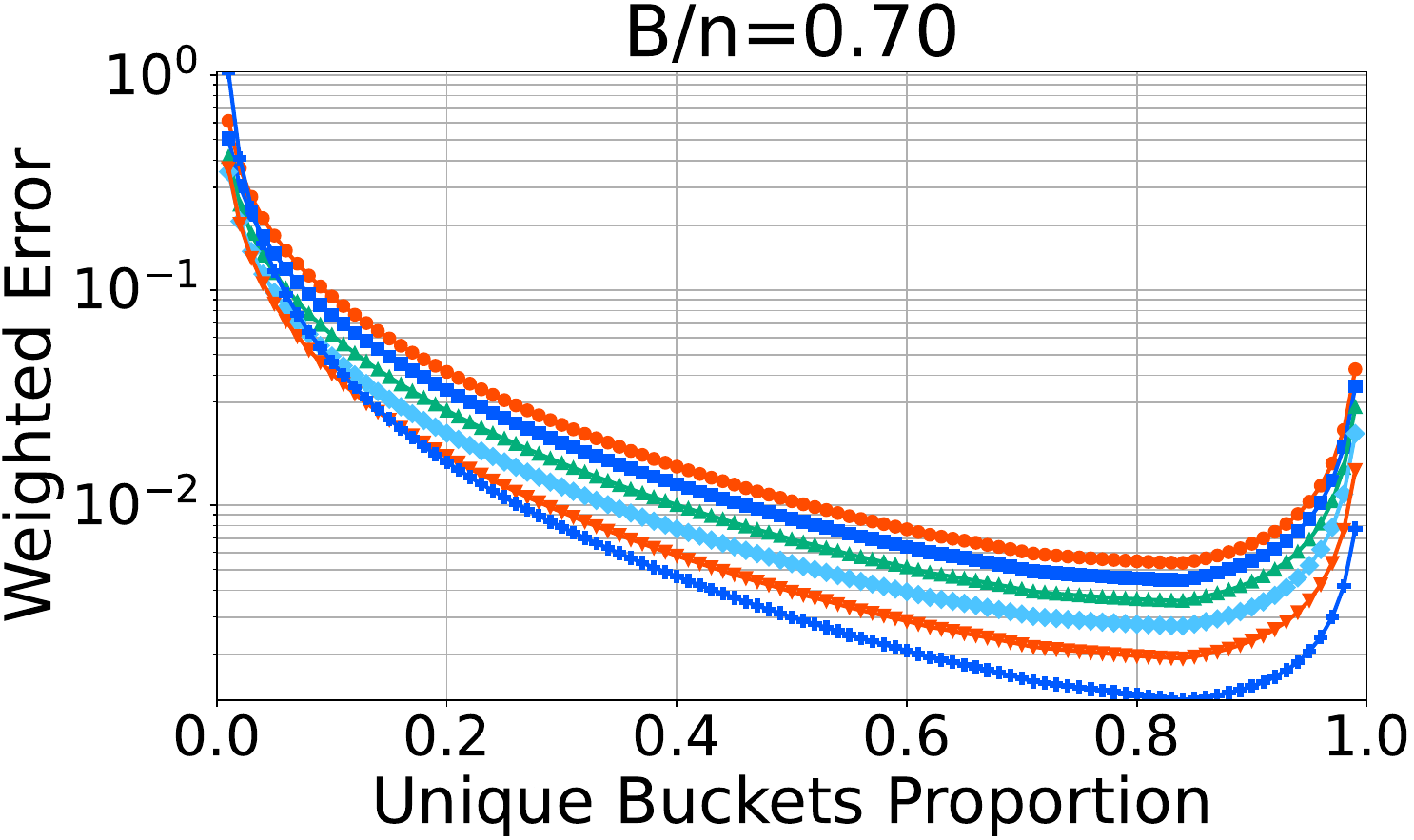}
    \end{minipage}
    \begin{minipage}[tb]{0.48\columnwidth}
        \centering
        \includegraphics[width=\columnwidth]{figures/lcms_error/google/space_0.8000.pdf}
    \end{minipage}
    \begin{minipage}[tb]{0.48\columnwidth}
        \centering
        \includegraphics[width=\columnwidth]{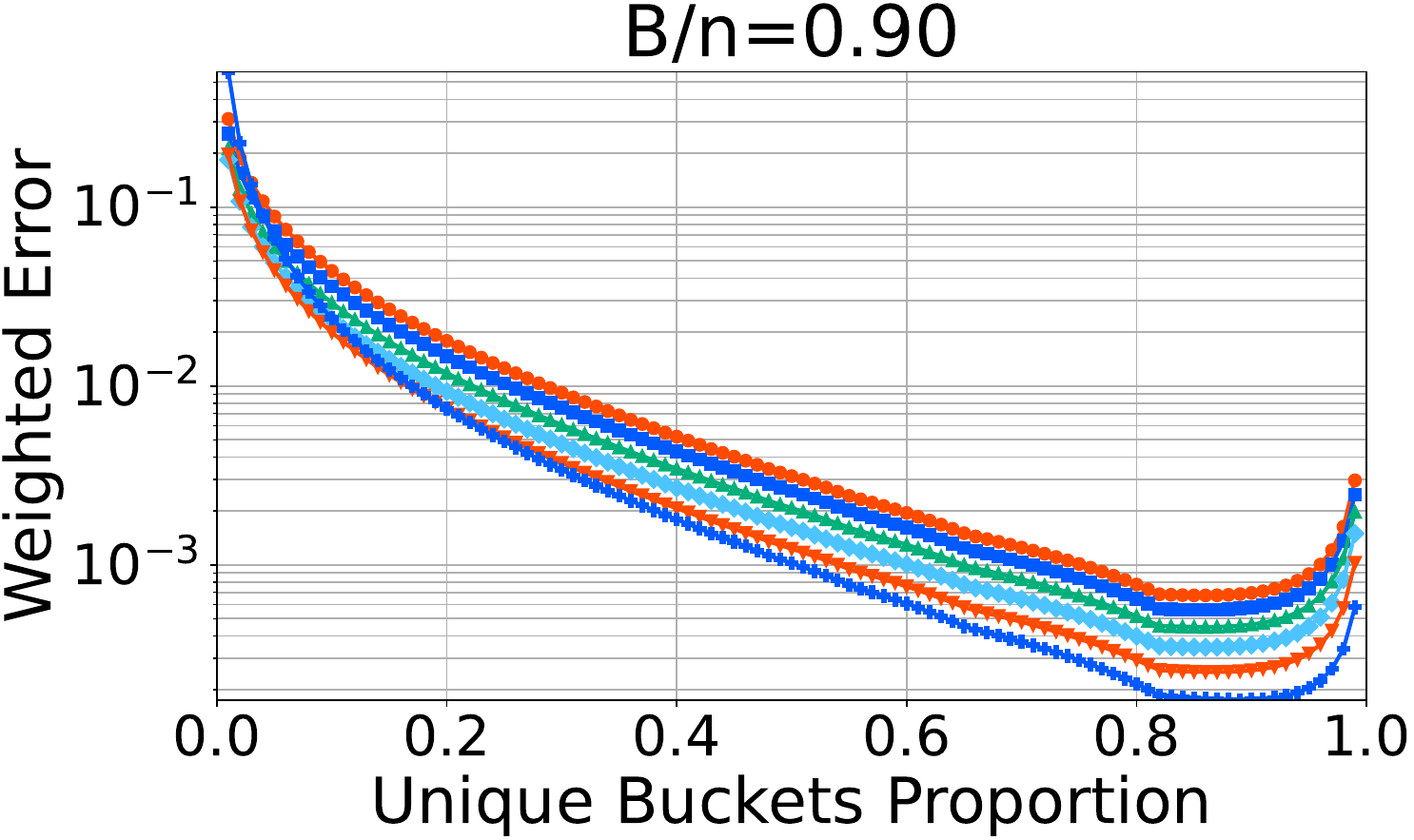}
    \end{minipage}
    \begin{minipage}[tb]{0.15\columnwidth}
        \centering
        \includegraphics[width=\columnwidth]{figures/lcms_error/google/legend.pdf}
    \end{minipage}
    \caption{The relationship between $2$ types of parameters and the weighted error on Google Books Ngram Viewer Datasets (Average over $10$ runs)}
    \label{fig:22}
\end{figure}

\end{document}